\documentclass[]{scrartcl}
\usepackage{amsmath,amsfonts,amsthm}
\usepackage{thmtools, thm-restate}
\usepackage{tikz}
\usetikzlibrary{matrix}
\usetikzlibrary{shapes}
\usetikzlibrary{arrows.meta}
\usetikzlibrary{shadows.blur}
\usetikzlibrary{calc}
\usetikzlibrary{positioning}

\usepackage[pdftex,colorlinks,linkcolor=black,urlcolor=black,citecolor=black]{hyperref}
\hypersetup{
	colorlinks=true,
    linkcolor={red!50!black},
    citecolor={blue!50!black},
    urlcolor={blue!80!black},
	bookmarksopen=true,
	bookmarksnumbered,
	bookmarksopenlevel=2,
	bookmarksdepth=3
}

\usepackage[ocgcolorlinks]{ocgx2}%
\usepackage[size=footnotesize]{todonotes}
\usepackage{hyperref}
\hypersetup{colorlinks=true,
            citecolor=green!40!black,
            linkcolor=red!20!black,
            urlcolor=blue!40!black,
            filecolor=cyan!30!black} 
\usepackage[capitalise,nameinlink]{cleveref}

\usepackage{authblk}

\newcommand{\email}[1]{
  \texttt{#1}
}

\usepackage[defaultlines=3,all]{nowidow}
\usepackage{microtype}

\newtheorem{theorem}{Theorem}
\newtheorem{lemma}[theorem]{Lemma}
\newtheorem{corollary}[theorem]{Corollary}

\newtheorem{proposition}[theorem]{Proposition}
\newtheorem{observation}[theorem]{Observation}
\newtheorem{claim}{Claim}

\usepackage[%
    backend=biber,%
    style=ieee,%
    sortcites,%
    dashed=false,%
    citestyle=numeric,%
    sorting=nyvt,%
    giveninits=true,%
    url=false,%
    doi=true,%
    isbn=false,%
    mincitenames=1,%
    maxcitenames=3,%
    maxbibnames=99,%
    defernumbers=true,%
    ]{biblatex}
\begin{document}

\date{}

\title{On constrained intersection representations of graphs and digraphs} 

\author[1]{Ferdinando Cicalese}
\affil[1]{Dept.\ of Computer Science, University of Verona,  Italy\protect\\\email{ferdinando.cicalese@univr.it}}
\author[2]{Cl\'ement Dallard}
\affil[2]{Dept.\ of Informatics, University of Fribourg, Switzerland\protect\\\email{clement.dallard@unifr.ch}}
\author[3]{Martin Milani\v{c}}
\affil[3]{FAMNIT and IAM, University of Primorska, Koper, Slovenia\protect\\\email{mailto:martin.milanic@upr.si}}

\makeatletter
\def\@maketitle{%
\newpage%
\null%
\begin{center}%
    \let\footnote\thanks %
    {\huge\bfseries \@title %
      \par
    }
  \vskip 1.5em
    {\lineskip .5em
     \begin{tabular}[t]{c}
        \baselineskip=12pt
        \@author
     \end{tabular}
     \par
    }
\end{center}
\par
\vskip 1.5em}
\makeatother

\newcommand\extrafootertext[1]{%
    \bgroup
    \renewcommand\thefootnote{\fnsymbol{footnote}}%
    \renewcommand\thempfootnote{\fnsymbol{mpfootnote}}%
    \footnotetext[0]{#1}%
    \egroup
}

\newcommand{\IN}{\textsf{in}}
\newcommand{\DIN}{\textsf{din}}
\newcommand{\WDIN}{\textsf{wdin}}
\newcommand{\UIN}{\textsf{uin}}
\newcommand{\tw}{\mathrm{tw}}
\newcommand{\itw}{\mathrm{itw}}
\newcommand{\mms}{\mathrm{mms}}
\newcommand{\tin}{\mathsf{tree} \textnormal{-} \alpha}
\newcommand{\h}{\eta}
\newcommand{\HH}{\mathcal{H}}
\newcommand{\G}{\mathcal{G}}
\newcommand{\N}{\mathbb{N}}
\renewcommand{\P}{\textsf{P}}
\newcommand{\NP}{\textsf{NP}}
\newcommand{\size}{s}
\renewcommand{\O}{\mathcal{O}}
\renewcommand{\deg}{\textsf{deg}}
\newcommand{\adeg}{\mathrm{\alpha\text{-}\deg}}

\maketitle

\extrafootertext{This work is supported in part by the Slovenian Research and Innovation Agency (I0-0035, research program P1-0285 and research projects J1-3003, J1-4008, J1-4084, J1-60012, and N1-0370) and by the research program CogniCom (0013103) at the University of Primorska.}
\extrafootertext{This paper was presented in part at the 33rd International Symposium on Algorithms and Computation (ISAAC 2022)\cite{ISAAC2022}.}

\begin{abstract}
We study the problem of determining optimal directed intersection representations of DAGs in a model introduced by Kostochka, Liu, Machado, and Milenkovic [ISIT2019]: vertices are assigned color sets so that there is an arc from a vertex $u$ to a vertex $v$ if and only if their color sets have nonempty intersection and $v$ gets assigned strictly more colors than $u$, and the goal is to minimize the total number of colors. 
We show that the problem is polynomially solvable in the class of triangle-free and Hamiltonian DAGs and also disclose the relationship of this problem with several other models of intersection representations of graphs and digraphs.
\end{abstract}

\section{Introduction}

\paragraph{Problem definition and state of the art.}
Given a digraph $D = (V,A)$, a \emph{directed intersection representation} of $D$ is a pair $(U,\varphi)$ where $U$ is a finite set of \emph{colors} and $\varphi$ is a \emph{proper coloring} of $D$, that is, a mapping assigning to each vertex $v\in V$ a set $\varphi(v)\subseteq U$ such that for any two vertices $u,v\in V$, it holds that 
\begin{equation*}
(u,v)\in A \quad\textrm{if and only if}\quad \varphi(u)\cap \varphi(v) \neq \emptyset\textrm{~and~}|\varphi(u)|<|\varphi(v)|\,.
\end{equation*}
The \emph{cardinality} of a directed intersection representation $(U,\varphi)$ of $D$ is defined as the number of colors, that is, $|U|$.
Note that if $(U,\varphi)$ is a directed intersection representation of a digraph $D$ and $W = (v_1,\ldots, v_k)$ is a walk in $D$, then $|\varphi(v_1)|<\ldots <|\varphi(v_k)|$, which implies that the directed graph $D$ is acyclic (that is, a \emph{DAG}).
On the other hand, Kostochka et al.~\cite{kostochka2019directed} showed that every DAG admits a directed intersection representation.
They initiated a study of the following invariant of DAGs.
The \emph{directed intersection number} (\emph{DIN}, for short) of a DAG $D = (V,A)$---denoted by $\DIN(D)$---is the smallest cardinality of a directed intersection representation of $D$.  
In \cite{kostochka2019directed,MR4231959}, the authors focused on characterizing the extremal values of DIN. 
They showed that: 
\begin{itemize} 
\item for every DAG $D$ with $n$ vertices, it holds that $\DIN(D) \leq \frac{5n^2}{8} -\frac{3n}{4}+1$;
\item for every $n$ there is a DAG $D$ with $n$ vertices such that
$\DIN(D)\ge \frac{9n^2}{16}-\frac{n}{4}-\frac{7}{4}$.
\end{itemize}
In \cite{COCOON2020,IWOCA2022}, Caucchiolo and Cicalese  studied the computational complexity of determining $\DIN(D)$. 
They showed that the problem of computing $\DIN(D)$ is $\NP$-hard even when $D$ is an arborescence (a tree with all the edges oriented away from the root). 
Moreover, for general DAGs and any $\epsilon>0$, the problem does not admit an $n^{1-\epsilon}$ approximation unless $\P = \NP$. 
Conversely, for the case of arborescences the problem is shown to be in APX, and in \cite{IWOCA2022} the authors also provide an asymptotic fully polynomial time approximation scheme. 

\paragraph{The main result.}
In this paper we continue the quest for islands of tractability in the complexity landscape for the problem of computing directed intersection representations of DAGs initiated in \cite{COCOON2020,IWOCA2022}. 
We focus on a class of graphs that, in the analysis of \cite{kostochka2019directed}, appears to include the instances that are the most demanding in terms of the DIN value. 
In fact, in \cite{kostochka2019directed}, a key point in the construction of the lower bound on the extremal value of $\DIN(D)$ is to consider DAGs that are both Hamiltonian and triangle-free.\footnote{Liu et al.~\cite{MR4231959} leave it as an open question whether for every $n$ the maximum value of $\DIN(D)$ among the DAGs with $n$ vertices is attained by one that is also Hamiltonian.} 
Here we show that for every $D$ in the class of Hamiltonian and triangle-free DAGs, computing the value of $\DIN(D)$ is a tractable problem solvable in time $\O(n^3)$.

A key element for obtaining such a result is 
the fact that for any Hamiltonian and triangle-free DAG $D$ it is possible to define a demand function $b$ on the vertices such that the value of $\DIN(D)$ can be exactly characterized in terms of the value of a maximum $b$-matching in the underlying graph of $D$, a parameter that can be computed in polynomial time~\cite{Gabow83,Pulleyblank73}.

\paragraph{Additional results, related problems, and literature.}
The proof of the above result leads us to introduce several generalizations and variants of intersection representations of 
graphs and digraphs (defined in \cref{subsec:INs}). 
We believe that these variants, for which we are able to prove interesting  results on their relations in terms of minimal intersection representations (summarized in \cref{Problems-relationships}), might be of independent interest.  

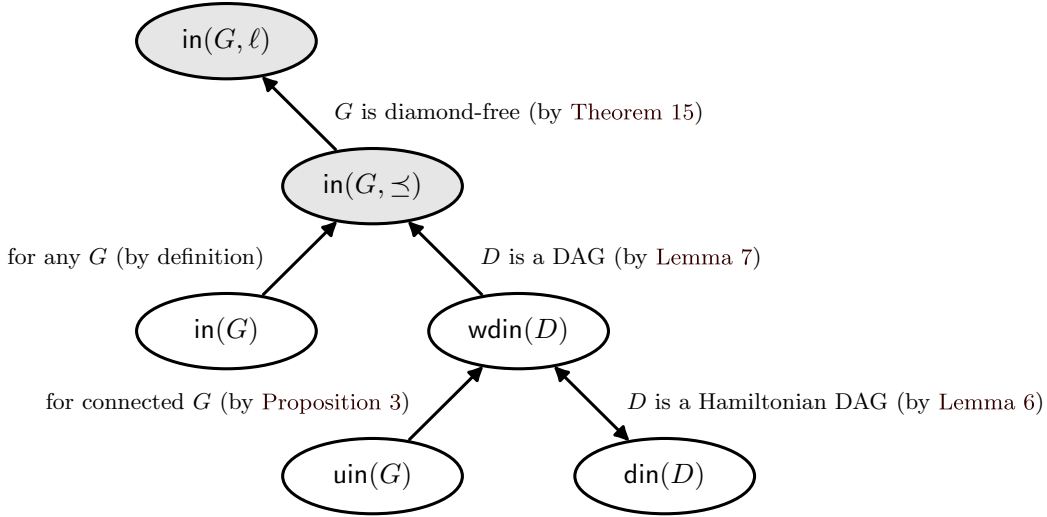
\begin{figure}[!ht]
    \centering
\resizebox{!}{.325\textheight}{
\begin{tikzpicture}[->=Latex,very thick,scale=0.8]
    \node[draw,ellipse,minimum height=25pt,minimum width=70pt,align=center,fill=black!10] (Gprec) at (0,0) {$\strut \IN(G,\preceq)$};
    \node[draw,ellipse,minimum height=25pt,minimum width=70pt,align=center,fill=black!10] (Gell) at ($(Gprec)+(-2.5,2.5)$) {$\strut \IN(G,\ell)$};
    \node[draw,ellipse,minimum height=25pt,minimum width=70pt,align=center] (G) at ($(Gprec)+(-2.5,-2.5)$) {$\strut \IN(G)$};
    \node[draw,ellipse,minimum height=25pt,minimum width=70pt,align=center] (Dprec) at ($(Gprec)+(2.5,-2.5)$) {$\strut \WDIN(D)$};
    \node[draw,ellipse,minimum height=25pt,minimum width=70pt,align=center] (D) at ($(Dprec)+(2.5,-2.5)$) {$\strut \DIN(D)$};
    \node[draw,ellipse,minimum height=25pt,minimum width=70pt,align=center] (UD) at ($(Dprec)+(-2.5,-2.5)$) {$\strut \UIN(G)$};
    
    \draw [-{Latex[round,length=2.5mm,width=2.5mm]}] (Gprec) to node[midway,right,draw=none,rectangle,outer sep=10pt] {\footnotesize $G$ is diamond-free (by \cref{thm:diamond-free})} (Gell);
    \draw [-{Latex[round,length=2.5mm,width=2.5mm]}] (G) to node[midway,left,draw=none,rectangle,outer sep=10pt] {\footnotesize for any $G$ (by definition)} (Gprec);
    \draw [-{Latex[round,length=2.5mm,width=2.5mm]}] (Dprec) to node[midway,right,draw=none,rectangle,outer sep=10pt] {\footnotesize $D$ is a DAG (by \cref{leqDIN-preceqIN})} (Gprec);
    \draw [{Latex[round,length=2.5mm,width=2.5mm]}-{Latex[round,length=2.5mm,width=2.5mm]}] (D) to node[midway,right,draw=none,rectangle,outer sep=10pt] {\footnotesize $D$ is a Hamiltonian DAG (by \cref{lemma:DAG-weakDAG-Hamiltonian})} (Dprec);
    \draw [-{Latex[round,length=2.5mm,width=2.5mm]}] (UD) to node[midway,left,draw=none,rectangle,outer sep=10pt] {\footnotesize for connected $G$ (by \cref{WDIN and UIN equal on connected graphs})} (Dprec);
\end{tikzpicture}
}
\caption{The relationships among the different types of intersection representations considered in this paper. 
The arc $\mathit{Prob}(a) \rightarrow \mathit{Prob}(b)$ is to be read ``$\mathit{Prob}(a)$ is a special case of $\mathit{Prob}(b)$.'' A label on an arc specifies the restriction on the class of instances for which the relation is proved to hold, along with a reference to the statement justifying the relation.
The highlighted parameters are introduced in this paper.}\label{Problems-relationships}
\end{figure}

It is known that any finite undirected graph $G$ admits an intersection representation given by a family of finite sets associated to its vertices, such that two vertices are adjacent if and only if their associated sets intersect. 
The minimum cardinality of the ground set of such a family is referred to as the \emph{intersection number} (\emph{IN}, for short) of the graph $G$ and denoted by $\IN(G)$.
Erd\H{o}s, Goodman and P\'osa \cite{erdos1966} showed that $\IN(G)$ equals the minimum number of cliques needed to cover the edges of $G$, i.e., the size of a minimum edge clique cover of $G$.
Determining this value was proved to be $\NP$-hard in \cite{Orlin1977} (see also \cite{KouSW78}). 
By \cite{feige1998zero,LundY94}, neither of the two problems is approximable within a factor of $|V|^{\epsilon}$ for any $\epsilon > 0$ unless $\P = \NP$. 
Applications of intersection representations and clique covers are found in areas as diverse as computational geometry, matrix factorization, compiler optimization, applied statistics, resource allocations, etc; see, e.g., the survey papers \cite{pullman1983clique,MR556057,roberts1985applications}, and the comprehensive introduction of \cite{cygan2016known}. 

In this paper several new variants are considered which contribute to this rich literature and in particular to the approach of \cite{roberts1985applications} of studying constrained versions of intersection representation and their applications. 
\medskip

The following practical scenario can be modeled by a constrained variant of the intersection number of undirected graphs, which is considered in the series of reductions leading to the proof of our main result. 

There is a shared resource (for example, a wireless communication channel) and a set of participants who want to use the resource.
However, no two participants are willing to share the resource at the same time unless they need to do it for accomplishing a common task.
In the wireless communication example, you might imagine that using the channel at the same time means to be possibly eavesdropped on while you do need to share temporally at least once with whomever you want to communicate with.  
We say that two participants are \emph{compatible} (with each other) if they need to accomplish a common task.
We assume that the resource can be used for an arbitrary number of time periods, where in each time period only a subset of pairwise compatible participants can share the resource.
Furthermore, if a set $S$ of pairwise compatible participants shares the resource in a certain time period, then the common task of any pair of participants in $S$ can be carried out in this period.
Since the use of the resource is expensive, our goal is to design a schedule of assigning the participants to time slots for using the resource that minimizes the total number of temporal slots, such that all pairs of participants can accomplish their tasks without ever using a resource together with an incompatible participant.
This problem can be modeled as the problem of computing the intersection number of the compatibility graph.
If we also assume that every participant requests to take part in at least a certain number of slots, we obtain the $\ell$-constrained intersection number, where $\ell(v)$ is the desired lower bound on the number of slots for participant $v$ (see \cref{subsec:INs} for definitions).

\section{Notations and definitions}

We denote by $\mathbb{N}$ the set of all positive integers and by $\mathbb{Z}_+$ the set of all nonnegative integers.
All graphs in this paper are finite and simple (that is, without loops and multiple edges), but may be directed or undirected. 
We will use the term \emph{graph} to refer to an undirected graph and the term \emph{digraph} to refer to a directed graph.

\paragraph{Definitions for graphs.}
We use standard graph theory terminology, see, e.g., West~\cite{MR1367739}.
A \emph{graph} is a pair $G = (V,E)$ where $V=V(G)$ is a finite set of \emph{vertices} and $E=E(G)$ is a set of $2$-element subsets of $V$ called \emph{edges}. 
Two vertices $u$ and $v$ in a graph $G = (V,E)$ are \emph{adjacent} if $\{u,v\}\in E$.
A vertex in a graph is \emph{universal} if it is adjacent to all other vertices.
A graph is said to be \emph{nontrivial} if it contains more than one vertex.
A \emph{vertex cover} in $G$ is a set $C$ of vertices such that every edge has at least one endpoint in $C$.
Let $b:V\to \mathbb{Z}_+$ be a capacity function on the vertices of $G$. 
A $b$-matching of $G$ is a function $x: E \to \mathbb{Z}_+$ 
such that for each vertex $v$ it holds that 
$\sum_{e \in E_v} x(e) \leq b(v),$ where $E_v$ denotes the set of edges incident with $v$. 
A maximum weight $b$-matching of $G$ is a $b$-matching of $G$ such that the total weight $\sum_{e \in E} x(e)$ is maximum among all $b$-matchings of $G$.
We use $\nu(G,b)$ to denote the total weight of a maximum weight $b$-matching of $G.$
Pulleyblank~\cite{Pulleyblank73} (see also~\cite{MR1956926}) showed that given a graph $G=(V,E)$ and a capacity function $b:V\to \mathbb{Z}_+$, a maximum weight $b$-matching in $G$ can be computed in time $\O(B \cdot |V|^2)$ where $B = 1+\max_{v\in V}b(v)$.
If $B$ is superpolynomial in $|V|$ and/or for the case of $G$ being sparse, we can use instead the algorithm of Gabow~\cite{Gabow83} which runs in time $\mathcal{O}(|E|^2  \cdot \log |V| \cdot \log B)$.
Since we can choose the best of the two above options for the computation of a maximum weight $b$-matching, we have the following.

\begin{theorem}\label{max-weight-b-matching}
Given a graph $G=(V,E)$ and a capacity function $b:V\to \mathbb{Z}_+$, a maximum weight $b$-matching in $G$ can be computed in time
$\O\left(\min\{B \cdot |V|^2, |E|^2  \cdot \log |V| \cdot \log B\}\right)$, where $B = 1+\max_{v\in V}b(v)$.
\end{theorem}

The \emph{triangle} is the complete graph on three vertices and the \emph{diamond} is the graph obtained from the $4$-vertex complete graph by removing an edge.
A graph $G$ is said to be \emph{triangle-free} if no induced subgraph of $G$ is isomorphic to the triangle; diamond-free graphs are defined similarly.
A \emph{clique} in a graph $G$ is a set of pairwise adjacent vertices.
A clique is \emph{maximal} if it is not included in any larger clique.

\paragraph{Definitions for digraphs.}
A \emph{digraph} is a pair $D = (V,A)$ where $V=V(D)$ is a finite set of \emph{vertices} and $A=A(D)$ is a set of ordered pairs of vertices called \emph{arcs}. 
Two vertices $u$ and $v$ in a digraph $D = (V,A)$ are \emph{adjacent} if $(u,v)\in A$ or $(v,u) \in A$.
A \emph{walk} in a digraph $D$ is a sequence $(v_1,\ldots, v_k)$ of vertices of $D$ such that $(v_i,v_{i+1})\in A$ for all $i\in \{1,\ldots,k-1\}$.
A \emph{path} in $D$ is a walk in which all vertices are pairwise distinct.
Given a positive integer $k$, a \emph{cycle of length $k$} in $D$ is a path $(v_1,\ldots, v_k)$ such that $(v_k,v_1)\in A$.
Note that a cycle of length one consists of a vertex $v$ at which $D$ has a loop (that is, an arc of the form $(v,v)$), and a cycle of length two consists of a pair of vertices $u,v$ such that both arcs $(u,v)$ and $(v,u)$ exist.
A digraph is \emph{acyclic} if it contains no cycles; a directed acyclic graph is referred to as a \emph{DAG}.
A digraph is \emph{Hamiltonian} if it has a path containing every vertex.
For a digraph $D = (V,A)$, the \emph{underlying graph of $D$} is the undirected graph $U(D) = (V,E)$ in which two distinct vertices are adjacent if and only if they are adjacent in $D$.
A digraph is \emph{strongly connected} if for every two vertices $u$ and $v$ in $D$, there is a $u,v$-path in $D$.
A \emph{strongly connected component} of $D$ is a maximal strongly connected subgraph.
A digraph $D$ is said to be \emph{triangle-free} if the underlying graph of $D$ is triangle-free; diamond-free digraphs are defined similarly.

\paragraph{Common definitions for graphs and digraphs.}
Let $G$ be a graph or a digraph.
An \emph{independent set} in $G$ is a set of pairwise nonadjacent vertices.
Let $f$ be a function from $V(G)$ to $\mathbb{Z}_+$.
Given a set $S\subseteq V(G)$, we denote by $f(S)$ the value $\sum_{v \in S} f(v)$.
Let $\alpha(G,f)$ denote the maximum value of $f(S)$ over all independent sets $S$ in $G$.
We say that $G$ is \emph{bipartite} if it admits a \emph{bipartition}, that is, a partition of its vertex set into two (possibly empty) independent sets.
Given a vertex $u\in V(G)$, the \emph{degree} of $u$ in $G$, denoted by $\deg_G(u)$ (or simply by $\deg(u)$ when the graph or digraph is clear from the context), is the number of vertices $v\in V(G)$ such that $u$ and $v$ are adjacent in $G$.
Note that if $D$ is a digraph without cycles of length one or two (in particular, if $D$ is a DAG), then the vertex degrees are the same in $D$ and in the underlying graph $U(D)$.

A \emph{partially ordered set} (or: a \emph{poset}) is a pair $(V, \preceq)$ where $V$ is a finite set and $\preceq$ is a binary relation on $V$ that is reflexive, antisymmetric, and transitive.
We write $u\prec v$ if $u\preceq v$ and $u\neq v$.
A poset element $v\in V$ is \emph{minimal} if there is no element $u\in V\setminus\{v\}$ such that $u\prec v$.

\subsection{Intersection representations of graphs and digraphs} \label{subsec:INs}

In this section we introduce several variants of intersection representation for graphs and digraphs. 
They are used in the proof of our main result, which consists of a sequence of reductions from one variant to another.
For the sake of completeness, we also recall the definitions of the intersection representation of an undirected graph and directed intersection representation of a digraph.

An \textbf{intersection representation of an undirected graph} $G = (V,E)$ is a pair $(U,\varphi)$ where $U$ is a finite set and $\varphi$ is a mapping assigning to each vertex $v\in V$ a set $\varphi(v)\subseteq U$ such that for any two distinct vertices $u,v\in V$, it holds that $\{u,v\}\in E$ if and only if $\varphi(u)\cap \varphi(v) \neq \emptyset$.
The \emph{cardinality} of an intersection representation $(U,\varphi)$ of $G$ is defined as the cardinality of $U$.
The \emph{intersection number} of an undirected graph $G$, denoted by $\IN(G)$, is the smallest cardinality of an intersection representation of $G$.

Let $G = (V,E)$ be a graph and let $\ell:V\to \mathbb{Z}_+$ be a demand function on the vertices of $G$.
An \textbf{$\ell$-constrained intersection representation of an undirected graph} $G$ is an intersection representation $(U,\varphi)$ of $G$ such that $|\varphi(v)|\ge \ell(v)$ for all $v\in V$.
The \emph{$\ell$-constrained intersection number} of $G$, denoted by $\IN(G,\ell)$, is the smallest cardinality of an $\ell$-constrained intersection representation of $G$.

A \emph{partially ordered graph} is a pair $(G,\preceq)$, where $G=(V,E)$ is an undirected graph and $\preceq$ is a partial order on the vertex set of $G$ (that is, $(V,\preceq)$ is a poset).  
An \textbf{intersection representation of a partially ordered graph} $(G,\preceq)$ is a pair $(U,\varphi)$ where $U$ is a finite set of \emph{colors} and $\varphi$ is a \emph{proper coloring} of $(G, \preceq)$, that is, a mapping assigning to each vertex $v\in V$ a set $\varphi(v)\subseteq U$ such that for any two distinct vertices $u,v\in V$, it holds that
\begin{itemize}
    \item  $\{u,v\}\in E$ if and only if $\varphi(u)\cap \varphi(v) \neq \emptyset$;
    \item if $u \prec v$, then $|\varphi(u)| < 
    |\varphi(v)|$.
\end{itemize}
The \emph{cardinality} of an intersection representation $(U,\varphi)$ of a partially ordered graph $(G,\preceq)$ is defined as the cardinality of $U$.
The \emph{intersection number} of a partially ordered graph $(G,\preceq)$, denoted by $\IN(G, \preceq)$, is the smallest cardinality of an intersection representation of $(G,\preceq)$.

One can view the $\ell$-constrained intersection number of a graph $G$ and the intersection number of a partially ordered graph $(G,\preceq)$ as  constrained variants of the intersection number of the graph $G$, placing it in a general framework proposed by Roberts in~\cite{MR556057} along with numerous other variants of the intersection number studied in the literature.

When discussing algorithms on partially ordered graphs, we assume that a partially ordered graph $(G,\preceq)$ is represented with the adjacency lists of the graph $G$ and an arbitrary DAG $D$ on the same vertex set such that $u\preceq v$ if and only if there exists a directed $u,v$-path from $u$ to $v$ in $D$.

A \textbf{directed intersection representation of a DAG} $D$ is a pair $(U,\varphi)$ where $U$ is a finite set of \emph{colors} and $\varphi$ is a \emph{proper coloring} of $D$, that is, a mapping assigning to each vertex $v\in V$ a set $\varphi(v)\subseteq U$ such that for any two vertices $u,v\in V$, it holds that 
\begin{equation} \label{eq:DIN-def}
(u,v)\in A \quad\textrm{if and only if}\quad \varphi(u)\cap \varphi(v) \neq \emptyset\textrm{~and~}|\varphi(u)|<|\varphi(v)|\,.
\end{equation}
The \emph{cardinality} of a directed intersection representation $(U,\varphi)$ of $D$ is defined as the number of colors, that is, $|U|$.
The \emph{directed intersection number} of a DAG $D$, denoted by $\DIN(D)$, is the smallest cardinality of a directed intersection representation of $D$.

A \textbf{weak directed intersection representation of a digraph} $D = (V, A)$ is a pair $(U,\varphi)$ where $U$ is a finite set of \emph{colors} and $\varphi$ is a \emph{weak proper coloring} of $D$, that is, a mapping assigning to each vertex $v\in V$ a set $\varphi(v)\subseteq U$ such that for any two distinct vertices $u,v\in V$, it holds that $(u,v)\in A$ if and only if $\varphi(u)\cap \varphi(v) \neq \emptyset$ and $|\varphi(u)|\leq|\varphi(v)|$.
Note that, unlike the definition of a directed intersection representation given in \cref{eq:DIN-def}, the cardinality constraint is expressed as a weak inequality.
This allows for representation of digraphs that are not acyclic.
More precisely, the family of digraphs that admit a weak directed intersection representation is a common generalization of graphs and DAGs (see \cref{prop:weak-characterization}).
Let $D$ be a digraph that admits a weak directed intersection representation.
The \emph{cardinality} of a weak directed intersection representation $(U,\varphi)$ of $D$ is defined as the number of colors, that is, $|U|$.
The \emph{weak directed intersection number} of $D$, denoted by $\WDIN(D)$, is the smallest cardinality of a weak directed intersection representation of $D$.

\begin{proposition}\label{prop:weak-characterization}
Let $D$ be a digraph.
Then, $D$ admits a weak directed intersection representation if and only if each strongly connected component $C$ of $D$ is obtained from its underlying graph $U(C)$ by replacing each edge with a pair of oppositely oriented arcs.  
\end{proposition}

\begin{proof}
Suppose first that $D=(V,A)$ admits a weak directed intersection representation $(U,\varphi)$ and let $C$ be a strongly connected component of $D$.
Let $(u,v)\in A$ be an arc in $C$.
Then, $\varphi(u)\cap \varphi(v)\neq \emptyset$.
Furthermore, since $C$ is strongly connected, there is a $v,u$-path $P = (v_1,\ldots, v_k)$ in $C$, with $v = v_1$ and $u = v_k$.
Then $|\varphi(v)| = |\varphi(v_1)|\le \ldots \le |\varphi(v_k)| = |\varphi(u)|$.
Consequently, $\varphi(v)\cap \varphi(u)\neq \emptyset$ and $|\varphi(v)| \le |\varphi(u)|$, implying that $(v,u)$ is also an arc in $D$.
Hence, $C$ is obtained from its underlying graph $U(C)$ by replacing each edge with a pair of oppositely oriented arcs. 

For the converse direction, suppose that each strongly connected component $C$ of $D$ is obtained from its underlying graph $U(C)$ by replacing each edge with a pair of oppositely oriented arcs.
We construct a weak directed intersection representation $(U,\varphi)$ of $D$ using the following procedure.
\begin{enumerate}
\item For each strongly connected component $C$ of $D$, let $\varphi_C$ be an intersection representation of the graph $U(C)$.
We assume that if $C$ and $C'$ are two distinct strongly connected components, then $\varphi_C$ uses distinct colors from $\varphi_{C'}$.
\item\label{scc} For all $v\in V$, set $\varphi(u) = \varphi_C(v)$ where $C$ is the strongly connected component containing~$v$.
\item \label{cross-arcs} For every arc $a = (u,v)$ of $D$ such that $u$ and $v$ belong to different strongly connected components of $D$, we introduce a new color $c_a$ and add it to both $\varphi(u)$ and $\varphi(v)$.
\item \label{top-sort} Let $(C_1,\ldots, C_k)$ be an ordering of the strongly connected components of $D$ such that for all arcs $(u,v)$ of $D$ where $u \in V(C_i)$ and $v \in V(C_j)$, with $i \neq j$, we have $i < j$.
(Such an ordering is known to exist, since there cannot be a cycle including vertices from different strongly connected components.)
Let $n_1 = \max_{v\in V(C_1)}|\varphi(v)|$ and for all $i \in \{2,\ldots, k\}$, let $n_i = \max\{n_{i-1}+1,\max_{v\in V(C_i)}|\varphi(v)|\}$.
For each $i\in \{1,\ldots, k\}$ and each vertex $v\in V(C_i)$, we add $n_i-|\varphi(v)|$ new colors to $\varphi(v)$.
\end{enumerate}
We claim that $(U,\varphi)$ is a weak directed intersection representation of $D$.
Let $u$ and $v$ be two distinct vertices of $D$.
First, observe that if $u$ and $v$ are nonadjacent, then $\varphi(u)\cap \varphi(v) =\emptyset$.
Suppose now that $(u,v)\in A$.
If $u$ and $v$ belong to the same strongly connected component $C_i$ of $D$, then $\varphi(u)\cap \varphi(v) \neq \emptyset$ by step \ref{scc} and $|\varphi(u)| = |\varphi(v)| = n_i$ by step \ref{top-sort}.
If $u$ and $v$ belong to different strongly connected components $C_i$ and $C_j$ of $D$, respectively, then $i<j$ and $c_{a}\in \varphi(u)\cap \varphi(v)$, where $a = (u,v)$, by step \ref{cross-arcs}; furthermore, $|\varphi(u)| = n_i < n_j = |\varphi(v)|$ by step \ref{top-sort}.
This shows that $(U,\varphi)$ is a weak directed intersection representation of $D$, as claimed.
\end{proof}

\Cref{prop:weak-characterization} implies that every DAG admits a weak directed intersection representation.
Note, however, that not every directed intersection representation is also a weak directed intersection representation. Indeed, in a directed intersection representation a pair of nonadjacent vertices can have intersecting color sets of the same cardinality, which for a weak intersecting representation would imply the presence of both arcs between the two vertices.

The weak directed intersection number has been previously considered for DAGs in \cite{Zeggiotti21}, where it was shown that the problem of computing $\WDIN(D)$ is $\NP$-hard when $D$ is an arbitrary DAG but polynomially solvable if $D$ is an arborescence, which contrasts with the $\NP$-hardness of computing $\DIN(D)$ for arborescences.
Furthermore, if $G$ is a connected graph and $D$ is the digraph obtained from $G$ by replacing each edge with a pair of oppositely oriented arcs, then $\WDIN(D)$ equals the \emph{uniform intersection number} of $G$, denoted by $\UIN(G)$ and defined as the minimum cardinality of a set $U$ such that $G$ admits an intersection representation over the set $U$ such that all vertices are assigned sets with the same cardinality (see, e.g.,~\cite{MR1068506,MR1132935,MR1623031}).

\begin{proposition}\label{WDIN and UIN equal on connected graphs}
Let $G$ be a connected graph and let $D$ be a digraph obtained from $G$ by replacing each edge $\{u,v\}$ of $G$ with a pair of oppositely oriented arcs between $u$ and $v$.
Then, $\UIN(G) = \WDIN(D)$.
\end{proposition}

\begin{proof}
Let $V$ be the common vertex set of $D$ and $G$.
The fact that $G$ is connected and that for each arc $(u,v)$ of $D$ we have also the opposite arc implies that,
given a set $U$ and a mapping $\varphi$ assigning to each vertex $v\in V$ a~set $\varphi(v)\subseteq U$, the pair $(U,\varphi)$ is a weak directed intersection representation of $D$ if and only if $(U,\varphi)$ is an intersection representation of $G$ such that all vertices are assigned sets with the same cardinality. 
Hence, $\UIN(G) = \WDIN(D)$, as claimed.
\end{proof}

In all of the above variants of intersection representations $(U,\varphi)$ of graphs, digraphs, and partially ordered graphs, we will refer to a representation $(U,\varphi)$ as \emph{optimal} if it has minimum cardinality.

\section{On the DIN of triangle-free Hamiltonian DAGs}

In this section, we present our main result about the polynomial computation of $\DIN(D)$ (including a corresponding optimal directed intersection representation
$(U, \varphi)$ for $D$) for a triangle-free Hamiltonian DAG $D$.

\begin{observation}\label{Ham-path-DAG-unique}
    Let $D$ be a Hamiltonian DAG.
    Then $D$ admits a unique Hamiltonian path.
\end{observation}

Following \cref{Ham-path-DAG-unique}, we may thus refer to \emph{the} Hamiltonian path of a DAG.

We prove the following theorem.

\begin{restatable}{theorem}{mainresult}\label{triangle-free-Hamiltonian-DIN-main}
\begin{sloppypar}
Let $D = (V,A)$ be a triangle-free Hamiltonian DAG and let $G$ be the underlying graph of $D$.
Let $P = (v_1,\ldots, v_n)$ be the Hamiltonian path in $D$ and let $b:V\to \mathbb{Z}_+$ be a capacity function on the vertices of $D$ defined as follows: $b(v_1) = 0$ and \hbox{$b(v_i) = \max\{b(v_{i-1})+\deg(v_{i-1})-\deg(v_{i})+1,0\}$, for all $i\in \{2,\ldots,n\}$.}
Then, \hbox{$\DIN(D) = |A| + b(V) - \nu(G,b)$} and an optimal directed intersection representation of $D$ can be computed in time $\mathcal{O}(|V|^3)$.
\end{sloppypar}
\end{restatable}

We give a concrete example of the application of this theorem in \cref{fig:example-computation}.
The graph used belongs to the class of {\em augmented source arc-path graphs} introduced in Liu et al.~\cite{MR4231959} as a 
candidate family of graphs for attaining the maximum possible DIN value among all graphs with the same number of vertices.

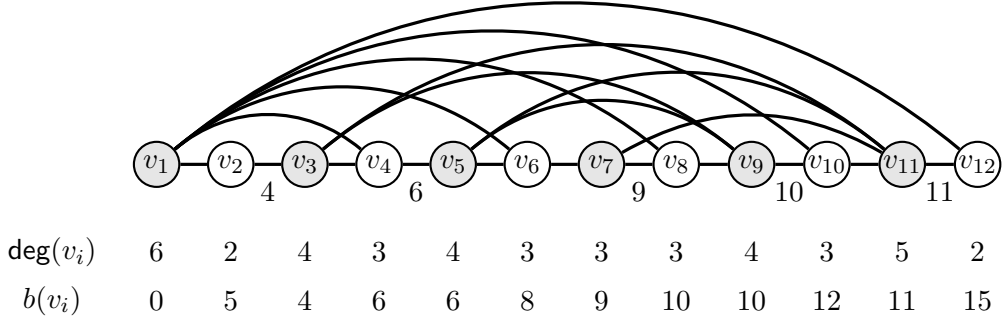
\begin{figure}[htp]
\centering%
\begin{tikzpicture}[]
    \tikzset{vertex/.style={draw=black,circle,fill=none,inner sep=0pt,minimum size=17pt,thick}}
    \begin{scope}
    \matrix [matrix of math nodes, nodes={anchor=center}, column sep=0.35cm,
    row 1/.style={nodes={vertex}}] (array) {
        & |[fill=black!10]| v_1 & v_2 & |[fill=black!10]| v_3 & v_4 & |[fill=black!10]| v_5 & v_6 & |[fill=black!10]| v_7 & v_8 & |[fill=black!10]| v_9 & v_{10} & |[fill=black!10]| v_{11} & v_{12} \\[15pt]
        \deg(v_i) & 6 & 2 & 4 & 3 & 4 & 3 & 3 & 3 & 4 & 3 & 5 & 2  \\
        b(v_i) & 0 & 5 & 4 & 6 & 6 & 8 & 9 & 10 & 10 & 12 & 11 & 15 \\
    };

    \draw [very thick] (array-1-2) to node[midway,below,draw=none,rectangle] {\strut} (array-1-3);
    \draw [very thick] (array-1-3) to node[midway,below,draw=none,rectangle] {\strut 4} (array-1-4);
    \draw [very thick] (array-1-4) to node[midway,below,draw=none,rectangle] {\strut} (array-1-5);
    \draw [very thick] (array-1-5) to node[midway,below,draw=none,rectangle] {\strut 6} (array-1-6);
    \draw [very thick] (array-1-6) to node[midway,below,draw=none,rectangle] {\strut} (array-1-7);
    \draw [very thick] (array-1-7) to node[midway,below,draw=none,rectangle] {\strut} (array-1-8);
    \draw [very thick] (array-1-8) to node[midway,below,draw=none,rectangle] {\strut 9} (array-1-9);
    \draw [very thick] (array-1-9) to node[midway,below,draw=none,rectangle] {\strut} (array-1-10);
    \draw [very thick] (array-1-10) to node[midway,below,draw=none,rectangle] {\strut 10} (array-1-11);
    \draw [very thick] (array-1-11) to node[midway,below,draw=none,rectangle] {\strut} (array-1-12);
    \draw [very thick] (array-1-12) to node[midway,below,draw=none,rectangle] {\strut 11} (array-1-13);    
    
    \foreach \i in {4,6,8,10,12}{
        \draw [very thick] (array-1-2) to[bend left=40] (array-1-\the\numexpr\i+1);
    }
    \foreach \i in {9,11}{
        \draw [very thick] (array-1-4) to[bend left=40] (array-1-\the\numexpr\i+1);
    }
    \foreach \i in {9,11}{
        \draw [very thick] (array-1-6) to[bend left=40] (array-1-\the\numexpr\i+1);
    }
    \draw [very thick] (array-1-8) to[bend left=30] (array-1-12);
    \end{scope}    
\end{tikzpicture}
\caption{An example of application of \cref{triangle-free-Hamiltonian-DIN-main} on a digraph $D = (V,A)$. 
The figure shows the underlying graph $G=(V,E)$; the digraph $D$ is obtained by orienting all the edges of $G$ from the left to the right. 
We explicitly give the parameters used in \cref{triangle-free-Hamiltonian-DIN-main}: degrees, capacity function $b()$, and edge values of an optimal $b$-matching (we only display the non-zero values).
The fact that the $b$-matching is indeed optimal can either be verified by observing that its total weight, $40$, matches the upper bound given by the $b$-weight of some vertex cover of $G$ (consider, e.g., $\{v_1,v_3,v_5,v_7,v_9,v_{11}\}$), or by using the fact that $\nu(G,b) = |A| + b(V) - \DIN(D) = 40$, since $\DIN(D) = 77$ (as shown in~\cite{MR4231959}), $|A| = 21$, and $b(V) = 96$.
\label{fig:example-computation}}
\end{figure}

\subsection{Proof of \cref{triangle-free-Hamiltonian-DIN-main}}
The proof of \cref{triangle-free-Hamiltonian-DIN-main} involves several intermediate steps that allow us to 
state relationships among the variants of intersection representations introduced above. 
Here, we present these steps and the relationships they involve among variants of IN and DIN, and defer the proofs to the following subsections.

Our first lemma shows that in the case of Hamiltonian DAGs, the notions of weak directed intersection representation and directed intersection representation coincide.

\begin{restatable}{lemma}{dindinleq}\label{lemma:DAG-weakDAG-Hamiltonian}
Let $D = (V, A)$ be a DAG. 
Then, every weak directed intersection representation of $D$ is a directed intersection representation.
Moreover, if $D$ is Hamiltonian, then every directed intersection representation of $D$ is a weak directed intersection representation, and $\DIN(D) = \WDIN(D)$.
\end{restatable}

Next, we show that weak directed intersection representations of a DAG $D$ correspond to intersection representations of a partially ordered graph $(G, \preceq)$, where $G$ is the underlying graph of $D$ and the partial order $\preceq$ is defined by the reachability relation in~$D$.

\begin{restatable}{lemma}{dinleqin}\label{leqDIN-preceqIN}
Let $D = (V, A)$ be a DAG. 
Let $G = (V, E)$ be the underlying graph of $D$ and let $\preceq$ be the partial order on $V$ defined by setting $u\preceq v$ if and only if there exists a $u,v$-path in $D$.
Then, every weak directed intersection representation of the DAG $D$ is an intersection representation of the partially ordered graph $(G, \preceq)$, and vice versa.
In particular, $\WDIN(D) = \IN(G, \preceq)$.
\end{restatable}

Then, we show that for a triangle-free graph $G$, and any given partial order $\preceq$ on the vertices of $G$, the computation of an intersection representation for the partially ordered graph $(G, \preceq)$ can be reduced to the computation of an $\ell$-constrained intersection representation of $G$ for some polynomially computable demand function $\ell$.

\begin{restatable}{theorem}{potolbtri}\label{fromPO-to-LBtri-free-graphs} 
Let $(G,\preceq)$ be a partially ordered graph such that $G = (V,E)$ is triangle-free.
Then, the $\alpha$-ranking $\ell:V\to \mathbb{Z}_+$ of $(G,\preceq)$ can be computed in linear time, and every optimal $\ell$-constrained intersection representation $(U,\varphi)$ of $G$ can be transformed to an optimal $\ell$-constrained intersection representation $(U,\psi)$ of $G$ that is also an optimal intersection representation of $(G,\preceq)$.
In particular, we have $\IN(G, \preceq) = \IN(G,\ell)$.
Moreover, this transformation can be done in time $\mathcal{O}(|U| \cdot |V|+|E| \cdot \max_{v\in V}|\varphi(v)|)$.
\end{restatable}

In fact, we derive \cref{fromPO-to-LBtri-free-graphs} from a more general result for the class of diamond-free graphs (\cref{fromPO-to-LBdiam-free-graphs}).
The analogue of \cref{fromPO-to-LBtri-free-graphs} for diamond-free graphs is given as~\cref{thm:diamond-free}.

Next, we show how to compute an optimal $\ell$-constrained intersection representation of a triangle-free graph $G$.

\begin{restatable}{theorem}{trianglefreelconstrained}\label{triangle-free-l-constrainedIN} 
Let $G = (V,E)$ be a triangle-free graph and $\ell:V\to \mathbb{Z}_+$ be a demand function on the vertices of $G$.
Let $b:V\to \mathbb{Z}_+$ be a capacity function on the vertices of $G$ defined by setting 
$b(v) = \max\{\ell(v) - \deg(v), 0\}$ for all $v\in V$.
Then $\IN(G, \ell) = |E| + b(V) - \nu(G,b)$
and an optimal $\ell$-constrained intersection representation $(U,\varphi)$ of $G$ such that $|\varphi(v)|\le \max\{\ell(v),\deg(v)\}$, for all $v\in V$, can be computed in time $\O(|V| \cdot \log L+\sum_{v\in V}\ell(v)+\min\{L \cdot |V|^2, |E|^2 \cdot \log |V| \cdot \log L\})$, where $L = 1+\max_{v\in V}\ell(v)$.
\end{restatable}

As a consequence of \cref{fromPO-to-LBtri-free-graphs,triangle-free-l-constrainedIN}, we obtain a characterization of the intersection number of a partially ordered triangle-free graph, along with a polynomial-time algorithm for its computation. 

\begin{restatable}{theorem}{trianglegreepreceqin}\label{triangle-free-preceqIN}
Let $(G,\preceq)$ be a partially ordered graph such that $G = (V,E)$ is triangle-free.
Let $M$ be the set of minimal elements in the poset $(V, \preceq)$.
Let $b:V\to \mathbb{Z}_+$ be the capacity function on the vertices of $G$ defined by setting
\begin{equation}\label{definition-of-b}
b(v) = \begin{cases}
 0 & \textrm{if } v \in M\,,\\
  \max\left\{1-\deg(v)+\max\limits_{u\colon  u \prec v} (b(u)+\deg(u))\,,0\right\} &
  \textrm{if } v \not\in M\,.
\end{cases}
\end{equation}
Then $\IN(G, \preceq) = |E| + b(V) - \nu(G,b)$ and an optimal intersection representation of $(G,\preceq)$ can be computed in time $\mathcal{O}(|V|^3)$.
\end{restatable}

\begin{figure}[h!]
    \centering
\resizebox{!}{0.15\textwidth}{
\begin{tikzpicture}[->=Latex,very thick]
    \node[draw,ellipse,minimum height=25pt,minimum width=70pt,align=center] (INELL) at (0,0) {$\strut \IN(G,\ell)$};
    \node[draw,ellipse,minimum height=25pt,minimum width=70pt,align=center] (BM) at ($(INELL)+(4,0)$) {$\strut \nu(G,b)$};
    \node[draw,ellipse,minimum height=25pt,minimum width=70pt,align=center] (INPO) at ($(INELL)+(-4,0)$) {$\strut \IN(G,\preceq)$};
    \node[draw,ellipse,minimum height=25pt,minimum width=70pt,align=center] (WDIN) at ($(INPO)+(-4,0)$) {$\strut \WDIN(D)$};
    \node[draw,ellipse,minimum height=25pt,minimum width=70pt,align=center] (DIN) at ($(WDIN)+(-4,0)$) {$\strut \DIN(D)$};

    \draw [-{Latex[round,length=2.5mm,width=2.5mm]}] (DIN) to node[midway,below,draw=none,rectangle,outer sep=10pt] {\Cref{lemma:DAG-weakDAG-Hamiltonian}} (WDIN);
    \draw [-{Latex[round,length=2.5mm,width=2.5mm]}] (WDIN) to node[midway,below,draw=none,rectangle,outer sep=10pt] {\Cref{leqDIN-preceqIN}} (INPO);
    \draw [-{Latex[round,length=2.5mm,width=2.5mm]}] (INPO) to node[midway,below,draw=none,rectangle,outer sep=10pt] {\Cref{fromPO-to-LBtri-free-graphs}} (INELL);
    \draw [-{Latex[round,length=2.5mm,width=2.5mm]}] (INELL) to node[midway,below,draw=none,rectangle,outer sep=10pt] {\Cref{triangle-free-l-constrainedIN}} (BM);
    \draw [-{Latex[round,length=2.5mm,width=2.5mm]}] (INPO) to[bend left = 30] node[midway,above,draw=none,rectangle,outer sep=10pt] {\Cref{triangle-free-preceqIN}} (BM);
\end{tikzpicture}
}
\caption{Steps of the proof of \cref{triangle-free-Hamiltonian-DIN-main}.}\label{chain-of-reductions}
\end{figure}
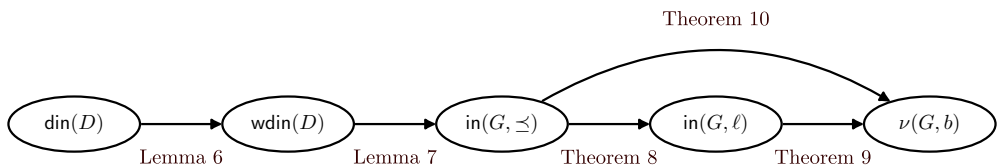
\cref{chain-of-reductions} summarizes the relationships established by the above steps.

By following the above reductions in the opposite order, the above lemmas and theorems lead to the computation of the directed intersection number of triangle-free Hamiltonian DAGs, i.e., it implies \cref{triangle-free-Hamiltonian-DIN-main}, which we recall here for convenience.

\mainresult*

\begin{proof}%
We first show that the following algorithm computes an optimal directed intersection representation of $D.$
\begin{enumerate}
    \item\label{graph} Compute the underlying graph $G = (V,E)$ of $D$.
    \item\label{demand} Let $\preceq$ be the partial order on $V$ such that $u\preceq v$ if and only if there exists a $u,v$-path in $D$.
    Using \cref{triangle-free-preceqIN}, we compute an optimal intersection representation $(U,\varphi)$ of $(G,\preceq)$.
\item\label{return} Return $(U,\varphi)$.
\end{enumerate}

\paragraph{Proof of correctness.}
By \cref{leqDIN-preceqIN}, $(U,\varphi)$ is an optimal weak directed intersection representation of $D$.
Consequently, by \cref{lemma:DAG-weakDAG-Hamiltonian}, $(U,\varphi)$ is an optimal directed intersection representation of $D$.
Hence, the algorithm is correct.

Let $M$ be the set of minimal elements in the poset $(V, \preceq)$, that is, $M = \{v_1\}$.
Let $b'$ be the capacity function on the vertices of $G$ defined in \cref{definition-of-b}, that is, 
\begin{equation*}
b'(v) = \begin{cases}
 0 & \textrm{if } v \in M\,,\\
  \max\left\{1-\deg(v)+\max\limits_{u\colon  u \prec v} \left(b'(u)+\deg(u)\right)\,,0\right\} &
  \textrm{if } v \not\in M\,.
\end{cases}
\end{equation*}
Consider an arbitrary $i\in \{2,\ldots, n\}$.
Since $v_{i-1}\prec v_i$, we have that 
$b'(v_{i})\ge 1-\deg(v_{i})+b'(v_{i-1})+\deg(v_{i-1})$, that is, $b'(v_{i})+\deg(v_i)> b'(v_{i-1})+\deg(v_{i-1})$.
Hence, the sequence $(b'(v_{j})+\deg(v_j))_{1\le j\le n}$ is increasing.
Since $\{u\colon u\prec v_i\} = \{v_1,\ldots, v_{i-1}\}$, this implies that $\max\{b'(u)+\deg(u)\colon u \prec v_i\}= b(v_{i-1})+\deg(v_{i-1})$.
It follows that the capacity function $b'$ is the same as the capacity function $b$ defined in the statement of the theorem.
Consequently, by \cref{triangle-free-preceqIN}, $\DIN(D) = |U| = \IN(G, \preceq) = |E| + b(V) - \nu(G,b)$.

\paragraph{Time complexity analysis.}
Step \ref{graph} can be executed in time $\O(|V|+|A|)$.
By \cref{triangle-free-preceqIN}, step \ref{demand} can be done in time $\O(|V|^3)$.
The time complexity of step \ref{return} is dominated by that of step \ref{demand}.
The claimed time complexity bound of $\O(|V|^3)$ follows.
\end{proof}

\subsection{On the DIN and weak DIN of Hamiltonian DAGs}

We first study the relationship between directed and weak directed intersection representations for arbitrary DAGs and show that they coincide for the case of Hamiltonian DAGs. 
This formalizes the relationships depicted by the two arcs in the bottom right corner of \cref{Problems-relationships}.

\dindinleq*
\begin{proof}
Let $(U, \varphi)$ be a weak directed intersection representation of $D$. 
Then, $(u,v) \in A$ if and only if $\varphi(u) \cap \varphi(v) \neq \emptyset$ and $|\varphi(u)| \leq |\varphi(v)|$.
Moreover, since $D$ is a DAG, whenever $(u,v) \in A$, there is no arc $(v,u)$, hence the inequality between the color sets' cardinalities must be strict, i.e., $|\varphi(u)| < |\varphi(v)|$.
If $(u,v)\not\in A$, then either $\varphi(u) \cap \varphi(v)=  \emptyset$ or $|\varphi(u)| > |\varphi(v)|$, which implies that either $\varphi(u) \cap \varphi(v)=  \emptyset$ or $|\varphi(u)| \ge |\varphi(v)|$.
Hence, $(U, \varphi)$ is also a directed intersection representation of $D$.

Suppose now that $D$ is Hamiltonian and let $(U, \varphi)$ be a directed intersection representation of $D$.
Clearly, we have that for each $(u,v) \in A$ the fact that $\varphi(u) \cap \varphi(v) \neq \emptyset \mbox{ and }|\varphi(u)| < |\varphi(v)|$ implies  that
\begin{equation}\label{arc}
\varphi(u) \cap \varphi(v) \neq \emptyset \mbox{ and }|\varphi(u)| \leq |\varphi(v)|.
\end{equation}
Moreover, the existence of a Hamiltonian path in $D$ implies that for each pair of vertices $u,v$ we have $|\varphi(u)|\neq |\varphi(v)|$. 
Therefore, if $(u,v) \not \in A$, since $(U,\varphi)$ is a directed intersection representation, at least one of the following must hold:
(i) $|\varphi(u)| \geq |\varphi(v)|$, and hence $|\varphi(u)|> |\varphi(v)|$; (ii) $\varphi(u) \cap \varphi(v) = \emptyset$.
The fact that \cref{arc} holds for each arc $(u,v)$ and (i) or (ii) holds whenever $(u,v) \not \in A$ implies that $(U,\varphi)$ is a weak directed intersection representation of $D$. 

The first part of the proof implies that $\DIN(D) \le \WDIN(D)$.
For Hamiltonian DAGs, the converse inequality is implied by the second part of the proof.
\end{proof}

\subsection{From DAGs to partially ordered graphs}

The next lemma formalizes the relationship depicted by the right arc in the middle of \cref{Problems-relationships}.

\dinleqin*
\begin{proof}
Let $(U, \varphi)$ be a weak directed intersection representation of $D$.
\begin{claim}
Fix distinct vertices $u,v$. If there is a $u,v$-path then $|\varphi(u)| < |\varphi(v)|$.
\end{claim}
\begin{proof}
Let $(u,v)$ be an arc of $D$. Then, we have $|\varphi(u)| \leq |\varphi(v)|$ and $\varphi(u) \cap \varphi(v) \neq \emptyset$. The latter, together with $(v,u) \not \in A$ (since $D$ is a DAG), implies  $|\varphi(u)| < |\varphi(v)|$. 
The claim now follows by transitivity, repeatedly using the above argument on the arcs of a $u,v$-path.
\end{proof}

Let us show that $(U,\varphi)$ satisfies the two properties defining an intersection representation of $(G, \preceq)$.
Let $u, v$ be a pair of distinct vertices in $G$.
\begin{enumerate}
\item Assume that $u \prec v$. 
By the definition of $\prec$ there is a $u,v$-path in $D$, hence by the claim above we have $|\varphi(u)| < |\varphi(v)|$. 
\item Let $\{u,v\} \in E$. 
Then, either $(u, v) \in A$ or $(v,u) \in A$, and in either case, $\varphi(u)\cap \varphi(v) \neq \emptyset$. 

Similarly, if $\varphi(u)\cap \varphi(v) \neq \emptyset$, then either $(u,v)$ or $(v,u)$ is an arc of $D$, which implies
that $\{u,v\} \in E$.
\end{enumerate}

Let us now assume that $(U, \varphi)$ is an intersection representation of $(G,\preceq)$. Then, for each pair of distinct vertices $u,v$ of $D$, we have the following.
\begin{itemize}
    \item If $(u,v) \in A$ then: (i) $\{u,v\} \in E;$ and 
    (ii) $u \prec v$. Hence, from (i) and (ii) respectively,
    $\varphi(u) \cap \varphi(v) \neq \emptyset$, and $|\varphi(u)| < |\varphi(v)|$.
    \item Assume now that: (i) $\varphi(u) \cap \varphi(v) \neq \emptyset$, and (ii) $|\varphi(u)| < |\varphi(v)|$. 
    From (i) we have $\{u,v\} \in E$. Since $G$ is the underlying graph of $D$ and $D$ is a DAG, it follows that exactly one of the arcs $(u,v), (v, u)$ is in $A$. 
    However, we cannot have $(v, u) \in A$, for otherwise $v \prec u$, which, together with $(U, \varphi)$ being an intersection representation of $(G,\preceq)$, would contradict the hypothesis $|\varphi(u)| < |\varphi(v)|$. Therefore, we must have $(u,v) \in A$.
\end{itemize}
The two items imply that $(U, \varphi)$ is a weak directed intersection representation of $D$ and conclude the proof. 
\end{proof}

\subsection{From partially ordered graphs to \texorpdfstring{$\ell$}{l}-constrained intersection representations}

In this section we prove \cref{fromPO-to-LBtri-free-graphs}, which gives the third step  in the sequence of reductions  among the different notions of intersection representations employed in the proof of \cref{triangle-free-Hamiltonian-DIN-main} (see \cref{chain-of-reductions}).
In fact, we are able to show that the relationship between the intersection representation of a partially ordered graph $(G, \preceq)$ and an $\ell$-constrained intersection representation of $G$, stated in \cref{fromPO-to-LBtri-free-graphs}, holds in the larger class of diamond-free graphs, as 
proved in \cref{fromPO-to-LBdiam-free-graphs}.
This, together with \cref{thm:diamond-free}, formalizes the relationship depicted by the topmost arc in \cref{Problems-relationships}.

Diamond-free graphs can be characterized as follows (see, e.g.,~\cite{zbMATH05880194,zbMATH06094028,zbMATH00022656}).

\begin{lemma}\label{lem:diamond-free-characterizations}
Let $G$ be a graph. Then, the following statements are equivalent.
\begin{enumerate}
    \item $G$ is diamond-free.
    \item Every edge of $G$ is contained in a unique maximal clique.
    \item For every vertex $v$ of $G$, the subgraph of $G$ induced by $N(v)$ is a disjoint union of complete graphs.
\end{enumerate}
\end{lemma}

For a vertex $v$ in a graph $G$, we denote by $\adeg(v)$ the \emph{independence degree} of $v$ in $G$, that is, the maximum cardinality of an independent set $S$ in $G$ such that $S\subseteq N(v)$, where $N(v)$ is the set of neighbors of $v$ in $G$.
Note that if $G$ is triangle-free, then every neighborhood of a vertex is an independent set, and the independence degree of a vertex coincides with its degree.
\cref{lem:diamond-free-characterizations} implies the following.

\begin{corollary}\label{cor:diamond-free-alpha-degree}
Given a diamond-free graph $G=(V,E)$ and a vertex $v$ of $G$, the independence degree of $v$ in $G$ can be computed in time  $\O(|V|+|E|)$.
\end{corollary}

Given a partially ordered graph $(G,\preceq)$, our approach is based on the following demand function on the vertices of a graph $G$ that depends on the independence degrees of the vertices of $G$ and the partial order given by $\preceq$.

Let $(G,\preceq)$ be a partially ordered graph with $G = (V,E)$. 
The \emph{$\alpha$-ranking of $(G,\preceq)$} is the demand function $\ell:V\to \mathbb{Z}_+$ on the vertices of $G$ such that 
\begin{equation}\label{definition-of-ell}
\ell(v) = \begin{cases}
 \adeg(v) & \textrm{if } v \in M\,,\\
  \max\left\{\adeg(v),1+\max\limits_{u\colon  u \prec v} \ell(u)\right\} &
  \textrm{if } v \not\in M\,,
\end{cases}
\end{equation}
where $M$ is the set of minimal elements in the poset $(V, \preceq)$.

We prove two lemmas, one for general partially ordered graphs and one for the diamond-free case.

\begin{lemma}\label{properties-of-alpha-ranking}
Let $(G,\preceq)$ be a partially ordered graph with $G = (V,E)$ and let $\ell$ be the $\alpha$-ranking of $(G,\preceq)$.
Then, $\ell(v)\le 2|V|$ for all $v\in V$, and every intersection representation of $(G,\preceq)$ is an $\ell$-constrained intersection representation of $G$.
In particular, $\IN(G,\ell)\le \IN(G, \preceq)$.
\end{lemma}

\begin{proof}
Note that if $D=(V,A)$ is any DAG representing the poset $(V,\preceq)$ and $u$ and $v$ are two distinct vertices in $D$ such that there is a $u,v$-path in $D$, then $\ell(u)<\ell(v)$.
This implies that
\begin{equation}\label{computation-of-ell}
\ell(v) =  \max\left\{\adeg(v),1+\max\limits_{u\colon  (u,v)\in A} \ell(u)\right\}, \text{~for all~} v \not\in M.
\end{equation}
Consequently, denoting by $\lambda$ the maximum length (that is, the number of arcs) in a path in $D$, for each vertex $v\in V$, we have $\ell(v)\le \max_{v\in V}\adeg(v)+\lambda\le 2|V|$.

Let $(U, \varphi)$ be an intersection representation of $(G,\preceq)$. 
To show that $(U, \varphi)$ is an \hbox{$\ell$-constrained} intersection representation of $G$, we need to show that for each vertex $v\in V$, we have $|\varphi(v)|\ge \ell(v)$. 
Letting $S$ be a maximum independent set in the subgraph of $G$ induced by $N(v)$, by the definition of intersection representation, the color sets of vertices in $S$ are pairwise disjoint. 
On the other hand, each one of these color sets has a nonempty intersection with $\varphi(v)$. 
It follows that $|\varphi(v)| \geq |S| = \adeg(v)$.
Moreover, by the definition of an intersection representation of a partially ordered graph, we have that $|\varphi(v)| > |\varphi(u)|$, for each $u \prec v$.
Therefore, $\varphi(v) \ge 1+\max\{\varphi(u)\colon (u,v)\in A\}$ if $v\not\in M$.
By induction on the length of a longest path in $D$ ending in $v$, it follows that $|\varphi(v)|\ge \ell(v)$.
Thus, $(U, \varphi)$ is an $\ell$-constrained intersection representation of $G$.
\end{proof}

\begin{lemma}\label{fromPO-to-LBdiam-free-graphs}
Let $(G,\preceq)$ be a partially ordered graph such that $G = (V,E)$ is diamond-free.
Let $\ell$ be the $\alpha$-ranking of $(G,\preceq)$.
Then, every optimal \hbox{$\ell$-constrained} intersection representation $(U,\varphi)$ of $G$ can be transformed to an optimal $\ell$-constrained intersection representation $(U,\psi)$ of $G$ that is also an optimal intersection representation of $(G,\preceq)$.
In particular, we have $\IN(G, \preceq) = \IN(G,\ell)$.
Moreover, this transformation can be done in time \hbox{$\O(|U|\cdot|V|+|E|\cdot\max_{v\in V}|\varphi(v)|+\mu(|V|,|E|))$}, where $\mu(|V|,|E|)$ denotes the time complexity of computing the set of maximal cliques of $G$.
\end{lemma}

\begin{proof}
It suffices to show that every optimal $\ell$-constrained intersection representation $(U,\varphi)$ of $G$ can be transformed in polynomial time to an optimal $\ell$-constrained intersection representation $(U,\psi)$ of $G$ that is also an intersection representation of $(G,\preceq)$.
This indeed suffices, since by \cref{properties-of-alpha-ranking}, any intersection representation of $(G, \preceq)$ with smaller cardinality than $(U,\psi)$ would also be an $\ell$-constrained intersection representation of $G$, contradicting the optimality of $(U,\varphi)$.

Let $(U, \varphi)$ be an optimal $\ell$-constrained intersection representation of $G$.
Let us denote by $\mathcal{K}$ the set of all maximal cliques $K$ in $G$ such that $|K|\ge 2$.
We first modify the representation $(U, \varphi)$ by performing the following procedure for each  maximal clique $K\in\mathcal{K}$.
We arbitrarily select two distinct vertices $u$ and $v$ in $K$.
Since $\{u,v\}\in E$, we have $\varphi(u)\cap \varphi(v) \neq \emptyset$; in particular, there exists a color $c_K\in \varphi(u)\cap \varphi(v)$.
Note that if $z$ is a vertex in $G$ such that $c_K\in \varphi(z)$, then $z\in K$, since otherwise $\{u,v,z\}$ would be included in a maximal clique other than $K$, contradicting the fact that, by \cref{lem:diamond-free-characterizations}, $K$ is the only maximal clique containing the edge $\{u,v\}$.
Hence, if for all $z\in K\setminus \{u,v\}$ such that $c_K\not\in \varphi(z)$, we add the color $c_K$ to the set $\varphi(z),$ we obtain an optimal $\ell$-constrained intersection representation of $G$.
After all maximal cliques of $G$ are processed this way, the resulting representation $(U, \varphi)$ is an optimal $\ell$-constrained intersection representation of $G$ such that for each maximal clique $K$ with $|K|\ge 2$ there is a color $c_K$ that is shared by color sets of all members of $K$.
Furthermore, for any two distinct maximal cliques $K$ and $K'$ in  $\mathcal{K}$, we have that $c_K\neq c_{K'}$.

We now show how to associate to each vertex $v\in V$ a subset $\psi(v)\subseteq \varphi(v)$ so that the resulting pair $(U,\psi)$ is an optimal $\ell$-constrained intersection representation of $G$ that is also an intersection representation of $(G, \preceq)$.
Let us write $U_\mathcal{K} = \{c_K\colon K\in \mathcal{K}\}$.
Note that for each vertex $v\in V$, we have that $|\varphi(v)\cap U_\mathcal{K}| = \adeg(v)$ and, hence, $|\varphi(v)\setminus U_\mathcal{K}| = |\varphi(v)|-\adeg(v)\ge \ell(v)-\adeg(v)$.
For each $v\in V$ such that $|\varphi(v)| >  \ell(v)$, we select an arbitrary set $X_v\subseteq \varphi(v)\setminus U_\mathcal{K}$ such that $|X_v| =  \ell(v)-\adeg(v)$.
Then, we define, for each $v\in V$
\[\psi(v) = \left.
  \begin{cases}
    (\varphi(v)\cap U_\mathcal{K})\cup X_v, & \text{if $|\varphi(v)| >  \ell(v)$\,,}\\
    \varphi(v), & \text{otherwise.}
\end{cases}
\right.
\]
Note that for all $v\in V$, we have that $\psi(v)\subseteq \varphi(v)$ and $|\psi(v)| = \ell(v)$.
The recursive definition of the function $\ell$ implies that $\ell(u) <  \ell(v)$ whenever $u\prec v$, and, hence, $|\psi(u)| < |\psi(v)|$ whenever $u \prec v$.
Next, we show that $(U,\psi)$ is an intersection representation of $(G, \preceq)$.
To this end, it remains to show that for any two distinct vertices $u$ and $v$ of $G$, we have $\{u,v\} \in E$ if and only if $\psi(u) \cap \psi(v) \neq \emptyset$.
Suppose first that $\{u,v\}= e \in E$. 
Then there exists a unique maximal clique $K\in \mathcal{K}$ such that $e\subseteq K$, and $c_K \in \varphi(u) \cap \varphi(v) \cap U_\mathcal{K} \subseteq \psi(u) \cap \psi(v)$, since we obtained the mapping $\psi$ from $\varphi$ by only removing colors not in $U_\mathcal{K}$.
Suppose now that $\{u,v\} \not \in E$.
Then $\psi(u) \cap \psi(v) = \emptyset$ follows directly from the relations $\psi(u) \subseteq \varphi(u)$ and $\psi(v) \subseteq \varphi(v)$, using the fact that $\varphi(u) \cap \varphi(v) = \emptyset$.
This shows that $(U,\psi)$ is an intersection representation of $(G, \preceq)$. 
Furthermore, since $|\psi(v)|= \ell(v)$, for all $v\in V$, and $|U| = \IN(G,\ell)$, the pair $(U,\psi)$ is an optimal $\ell$-constrained intersection representation of $G$. 

To conclude the proof, we show that $\psi$ can be computed within the stated time complexity bound.
We can assume that the set $U$ of colors is of the form $\{1,\ldots, |U|\}$. 
Then, for each $v\in V$, we can sort the set $\varphi(v)$ in time $\mathcal{O}(|U|)$.
By assumption, the set $\mathcal{K}$ of maximal cliques with at least two vertices can be computed in time $\O(\mu(|V|,|E|))$.
Having sorted the sets, for each clique $K\in \mathcal{K}$, we can compute $c_K$ in time $\mathcal{O}(|\varphi(u)|+|\varphi(v)|)$
where $u$ and $v$ are two arbitrary vertices in $K$; in total, this takes time $\mathcal{O}(|\mathcal{K}|\max_{v\in V}|\varphi(v)|)$. 
Furthermore, for each vertex $v\in V$ we can assure that for each of the $\adeg(v)$ cliques $K\in \mathcal{K}$ containing $v$, the color $c_K$ is contained in $\varphi(v)$ in time $\O(|U|+\adeg(v))$, which simplifies to $\O(|U|)$ since $\adeg(v)\le |U|$; in total, this takes time $\mathcal{O}(|U|\cdot|V|)$.
Next, for each $v\in V$, we compute the set $\psi(v)$ using its definition in time $\mathcal{O}(|U|+\adeg(v))$; in total, this again takes time $\mathcal{O}(|U|\cdot|V|)$.
Note that \cref{lem:diamond-free-characterizations} implies that $|\mathcal{K}| \leq |E|$, since each edge belongs to exactly one clique in $\mathcal{K}$.
Therefore, the total time complexity is $\O(|U|\cdot|V|+|E|\cdot\max_{v\in V}|\varphi(v)|+\mu(|V|,|E|))$, as claimed.
\end{proof}

\Cref{fromPO-to-LBdiam-free-graphs} implies the following two results, highlighting a slightly different time complexity depending on whether the diamond-free graph $G$ is also triangle-free.
For the general diamond-free case, we obtain the following.

\begin{theorem}\label{thm:diamond-free}
    Let $(G,\preceq)$ be a partially ordered graph such that $G = (V,E)$ is diamond-free.
    Then, the $\alpha$-ranking $\ell:V\to \mathbb{Z}_+$ of $(G,\preceq)$ can be computed in time $\O(|V|\cdot (|V|+|E|))$, and every optimal \hbox{$\ell$-constrained} intersection representation $(U,\varphi)$ of $G$ can be transformed to an optimal $\ell$-constrained intersection representation $(U,\psi)$ of $G$ that is also an optimal intersection representation of $(G,\preceq)$.
    In particular,  we have $\IN(G, \preceq) = \IN(G,\ell)$.
    Moreover, this transformation can be done in time \hbox{$\O(|V|^3\log{|V|}+|U|\cdot|V|+|E|\cdot\max_{v\in V}|\varphi(v)|)$}.
\end{theorem}

\begin{proof}
Let $D=(V,A)$ be any DAG representing the poset $(V,\preceq)$.
Recall from the proof of \cref{properties-of-alpha-ranking} that if $v$ is a non-minimal element of the poset $(V,\preceq)$, then the $\alpha$-ranking $\ell$ can be evaluated on $v$ using \cref{computation-of-ell}.
Hence, by traversing the elements of the poset $(V,\preceq)$ following an arbitrary topological sort of $D$ and applying \cref{cor:diamond-free-alpha-degree}, the function $\ell$ can be computed in time $\O(|V|\cdot (|V|+|E|))$.

By \cref{fromPO-to-LBdiam-free-graphs}, to complete the proof, it suffices to show that if $G$ is diamond-free, then one can compute the set $\mathcal{C}(G)$ of maximal cliques of $G$ in time $\O(|V|^3\log |V|)$.
    
    For each vertex $v$ in $G$, by \cref{lem:diamond-free-characterizations}, the subgraph of $G$ induced by $N(v)$ is a disjoint union of complete graphs.
    Hence, for each vertex $v$, we can compute in time $\O(|V|+|E|)$ the set $\mathcal{C}_v$ of $\O(|V|)$ maximal cliques containing $v$.
    The union of these sets, $\bigcup_{v\in V}\mathcal{C}_v$, contains a total of $\O(|V|^2)$ sets, but may contain repetitions.
    We then lexicographically sort the elements of this multiset by performing $\O(|V|^2\log (|V|^2))$ pairwise comparisons, each of which can be done in time $\O(|V|)$ (assuming that the sets are already sorted with respect to some fixed vertex ordering of $V$, which can be achieved, if necessary, with a linear-time preprocessing step).
    Having lexicographically sorted the multiset $\bigcup_{v\in V}\mathcal{C}_v$, we can eliminate repetitions in time $\O(|V|^3)$, thus obtaining the set $\mathcal{C}(G)$. 
\end{proof}

For triangle-free graphs, we obtain an improved time complexity given by~\Cref{fromPO-to-LBtri-free-graphs}, which we restate for convenience.

\potolbtri*

\begin{proof}
We again use \cref{fromPO-to-LBdiam-free-graphs}.
Since $G$ is triangle-free, $\adeg(v) = \deg(v)$, for all $v \in V$.
Thus, proceeding as in the proof of \cref{thm:diamond-free}, the $\alpha$-ranking $\ell$ of $(G,\preceq)$ can be computed in linear time.
Furthermore, following the notation of \cref{fromPO-to-LBdiam-free-graphs}, since $G$ is triangle-free, we have $\mu(|V|,|E|) = \mathcal{O}(|V|+|E|)$, since the set of maximal cliques of $G$ coincides with the set of its isolated vertices and its edges, which can be computed in linear time.
\end{proof}

We conclude this subsection by showing that the inequality given by \cref{properties-of-alpha-ranking}, which is satisfied with equality for the case of diamond-free graphs by \cref{fromPO-to-LBdiam-free-graphs}, can be strict in general.

\begin{proposition}\label{prop:strict}
There exists a partially ordered graph $(G,\preceq)$   such that $\IN(G,\ell)<\IN(G, \preceq)$, where $\ell$ is the $\alpha$-ranking of $(G,\preceq)$.
\end{proposition}

\begin{proof}
Let $G$ be a $20$-vertex graph whose vertex set admits a partition into $10$ parts $\{u_i,v_i\}$, $1\le i\le 10$, such that two distinct vertices of $G$ are adjacent if and only if they belong to different parts.
Note that the independence degree of every vertex of $G$ is equal to~$2$.
Let $\preceq$ be the poset on $V(G)$ in which there is only one pair of vertices in relation, say $x\prec y$ (where $x$ and $y$ are two arbitrary but fixed distinct vertices of $G$).
Then, the $\alpha$-ranking $\ell$ of $(G,\preceq)$ takes value $3$ on $y$ and $2$ on all other vertices.
Consider the following intersection representation $(U,\varphi)$ of $G$.
Let $U = \{1,\ldots,6\}$, assign to $u_1,\ldots, u_{10}$ the $10$ three-element subsets of $\{1,\ldots, 6\}$ containing $1$, and set $\varphi(v_i) = \{1,\ldots, 6\}\setminus \varphi(u_i)$, for all $i\in \{1,\ldots, 10\}$.
Then, $|\varphi(v)| = 3\ge \ell(v)$ for all vertices $v\in V(G)$, which implies that $(U,\varphi)$ is an $\ell$-constrained intersection representation of $G$.
Hence, $\IN(G,\ell) \le 6$.

So, to show that $\IN(G,\ell) < \IN(G, \preceq)$, it suffices to prove that we need more than $6$ colors in every intersection representation of $(G,\preceq)$.
Consider an optimal intersection representation $(U',\varphi')$ of $(G,\preceq)$, and suppose for a contradiction that $|U'| \le 6$.
For each $i\in \{1,\ldots, 10\}$, let $A_i = \varphi'(u_i)$ and $B_i = \varphi'(v_i)$.
Then $A_i\cap B_i = \emptyset$, for all $i\in \{1,\ldots, 10\}$.
Furthermore, for all $1\le i,j\le 10$ with $i\neq j$, we have $A_i\cap A_j\neq \emptyset$, $A_i\cap B_j\neq \emptyset$, $B_i\cap A_j\neq \emptyset$.
This implies that the sets $A_i$ and $A_j$ are incomparable: if, say, $A_i\subseteq A_j$, then $A_i$ would not intersect $B_j$. 
Similarly, the sets $A_i$ and $B_j$ are incomparable, since if, say, $A_i\subseteq B_j$, then $A_i$ would not intersect $A_j$. 
We infer that the $20$ sets $A_1,\ldots, A_{10},B_1,\ldots, B_{10}$ form an antichain in the partial order on the power set of $U'$ given by the inclusion relation.
By Sperner's theorem~\cite{zbMATH02575694}, any antichain $\mathcal{A}$ of subsets of $\{1,\ldots, n\}$ has size at most $\binom{n}{\left\lfloor n/2\right\rfloor}$, with equality if and only if $\mathcal{A}$ consists of all subsets of $\{1,\ldots, n\}$ that have size $\left\lfloor n/2\right\rfloor$ or all that have size $\left\lceil n/2\right\rceil$.
Since $\binom{6}{3} = 20$ and $\binom{n}{\left\lfloor n/2\right\rfloor}<20$ for $n<6$, we infer that $|U'| = 6$ and that the $20$-element set family $\{\varphi'(v)\colon v\in V(G)\}$ coincides with the set of all $3$-element subsets of $U'$.
In particular, $|\varphi'(x)| = |\varphi'(y)| = 3$, a contradiction with the fact that $x\prec y$ and $(U',\varphi')$ is an intersection representation of $(G,\preceq)$.
\end{proof}

\subsection{The \texorpdfstring{$\ell$}{l}-constrained intersection number of triangle-free graphs}

In this section, we prove \cref{triangle-free-l-constrainedIN}, our final ingredient needed for the proof of \cref{triangle-free-Hamiltonian-DIN-main} (see \cref{chain-of-reductions}).

\trianglefreelconstrained*

\begin{proof}
First, we prove that $|E| + b(V) - \nu(G,b)$ is an upper bound on the $\ell$-constrained intersection number of the graph $G$ 
by constructing an $\ell$-constrained intersection representation $(U,\varphi)$ with cardinality at most $|E| + b(V) - \nu(G,b)$. 
We associate to each edge $e\in E$ a unique color $c_e$ and define $U_E = \{c_e\colon e\in E\}$.
Let $x: E\to\mathbb{Z}_+$ be a maximum weight $b$-matching in $G$.
Then $\sum_{e\in E}x(e) = \nu(G,b)$.
We denote by $E_v$ the set of edges in $G$ incident with a vertex $v$.
To each edge $e$ of $G$, we associate another set $C_e$ of $x(e)$ new colors, 
and to each vertex $v$ of $G$, we associate a set $C_v$ of $b(v)-\sum_{e\in E_v}x(e)$ new colors, 
so that the sets $U_E$, $C_e$, $e\in E$, and $C_v$, $v\in V$, are pairwise disjoint.
Note that this construction is indeed possible: 
for each vertex $v$, the value of $b(v)-\sum_{e\in E_v}x(e)$ 
is a non-negative integer since $x$ is a $b$-matching in $G$.
We define
\begin{equation}\label{constructing-U-from-x}
U = U_E\cup \Bigg(\bigcup_{e\in E}C_e\Bigg) \cup \Bigg(\bigcup_{v\in V}C_v\Bigg)
\quad\textrm{and}\quad
\varphi(v) = \{c_e\colon e\in E_v\}\cup C_v\cup \left(\bigcup_{e\in E_v}C_e\right).
\end{equation}
Next, we show that $(U,\varphi)$ is an $\ell$-constrained
intersection representation of $G$.
Clearly $\varphi(v)\subseteq U$ for all $v\in V$.
Furthermore, for each vertex $v\in V$, by definition of $b$, it holds that $b(v) + \deg(v) \geq \ell(v)$, and we have
\begin{flalign}\label{vertex-bound}
|\varphi(v)| = |E_v|+|C_v|+\sum_{e\in E_v}|C_e|\nonumber
&=\deg(v)+b(v)-\sum_{e\in E_v}x(e)+\sum_{e\in E_v}x(e)\\
&=\deg(v)+b(v) \geq \ell(v)\,.
\end{flalign}
Therefore, the color sets assigned to the vertices by $\varphi$ satisfy the lower bounds defined by the demand function $\ell$.
Note also that $|\varphi(v)|= \deg(v)+b(v)\le \max\{\ell(v),\deg(v)\}$ for all $v\in V$.

We now prove that $(U, \varphi)$ is an intersection representation of $G$, by showing that for any distinct vertices $u$ and $v$ of $G$,
it holds that $\{u,v\}\in E$ if and only if $\varphi(u)\cap \varphi(v) \neq \emptyset$.

For this, assume first that $e = \{u,v\}\in E$.
Then $e\in E_{u}\cap E_{v}$ and thus $c_e\in \varphi(u) \cap \varphi(v)$, implying that $\varphi(u) \cap \varphi(v)\neq\emptyset$.
For the converse direction, assume that $\varphi(u)\cap \varphi(v) \neq \emptyset$.
Let $c\in \varphi(u)\cap \varphi(v)$.
Because for any two vertices $u,v\in V$, we have that $C_u\subseteq \varphi(v)$ if and only if $u = v$, the color $c$ cannot belong to any set $C_{z}$ for $z\in V$.
Therefore, since the sets $U_E$, and $C_e$, $e\in E$, are pairwise disjoint, the color $c$ must belong to either $U_E$ or to some set $C_e$ for $e\in E$. 
Assume first that $c\in U_E$.
Then $c = c_e$ for some $e\in E$, which implies that $e\in E_{u}\cap E_{v}$ and thus $e = \{u,v\}$, i.e., $\{u,v\}$ is an edge of $E$.
Assume now that $c\in C_e$ for some $e\in E$.
Since $c\in \varphi(u)\cap \varphi(v)$, we infer that $e\in E_{u}$ and $e\in E_{v}$, which again implies that $\{u,v\}$ is an edge of $G$.
This concludes the proof that $(U,\varphi)$ is an intersection representation of $G$. 

It remains to show that $|U| = |E| + b(V) - \nu(G,b)$.
Using the definition of $U$ we infer that, as claimed, its cardinality is 
\begin{eqnarray*}
|U| &=& |U_E|+\sum_{e\in E}|C_e|+\sum_{v\in V}|C_v| 
= |E|+\sum_{e\in E}x_e+\sum_{v\in V}\left(b(v)-\sum_{e\in E_v}x(e)\right)\\
&=& |E|+\nu(G,b)+\sum_{v\in V}b(v)-\sum_{v\in V}\sum_{e\in E_v}x(e)%
= |E|+\nu(G,b)+b(V)-2\sum_{e\in E}x(e)\\
&=& |E|+\nu(G,b)+b(V)-2\nu(G,b)
= |E|+b(V)-\nu(G,b).
\end{eqnarray*}

Second, we prove that $|E| + b(V) - \nu(G,b)$ is a lower bound for the $\ell$-constrained intersection number of $G$.
Let $(U,\varphi)$ be an arbitrary $\ell$-constrained intersection representation of $G$. 
We need to show that $|U|\ge |E| + b(V) - \nu(G,b)$.
For each edge $e = \{u,v\}\in E$, we have $\varphi(u)\cap \varphi(v) \neq \emptyset$; in particular, there exists a color $c_e\in \varphi(u)\cap \varphi(v)$.
As in the proof of \cref{fromPO-to-LBtri-free-graphs}, we have that for any two distinct edges $e,e'\in E$, we have $c_e\neq c_{e'}$.

Because of the triangle-free condition we have that the neighborhood of $v$ is an independent set, hence, the corresponding color sets are pairwise disjoint. 
On the other hand, each one of these color sets has a nonempty intersection with $\varphi(v)$. 
It follows that $|\varphi(v)| \geq \deg(v)$.
By the assumption that $(U, \varphi)$ is an $\ell$-constrained intersection representation of $G$, we also have $|\varphi(v)| \geq \ell(v)$ for all $v\in V$, and hence $|\varphi(v)| \geq \max\{\ell(v), \deg(v)\}$ for all $v\in V$.

Let $U_E = \{c_e\colon e\in E\}$.
For all $v\in V$, let us denote $\varphi'(v) = \varphi(v)\setminus U_E$.
It follows that for any vertex $v\in V$, we have that 
\[|\varphi'(v)| = |\varphi(v)|-\deg(v) \ge \max\{0, \ell(v)-\deg(v)\} = b(v).\]

\begin{claim}\label{claim-varphi-b}
We may assume without loss of generality that $|\varphi'(v)| = b(v)$, for all $v\in V$.
\end{claim}
\begin{proof}
Suppose that this is not the case.
Then, we choose for each $v\in V$ such that $|\varphi'(v)| > b(v)$, an arbitrary set $X_v\subseteq \varphi'(v)$ such that $|X_v| = b(v)$ and define, for all $v\in V$
\[\psi(v) = \left.
  \begin{cases}
    (\varphi(v)\cap U_E)\cup X_v, & \text{if $|\varphi'(v)| > b(v)$\,,}\\
    \varphi(v), & \text{otherwise.}
\end{cases}
\right.
\]
Similarly as above, let us denote $\psi'(v) = \psi(v)\setminus U_E$ for all $v\in V$.
By construction, we have $|\psi'(v)| = b(v)$, for all $v\in V$, and consequently $|\psi(v)|=|\psi'(v)|+\deg(v) \geq \ell(v)$ for all $v\in V$. 
Hence, $\psi$ satisfies the $\ell$-constraints on the vertices. 
To see that $(U,\psi)$ is indeed an $\ell$-constrained intersection representation of $G$, it remains to 
show that $\{u,v\} \in E$ if and only if 
$\psi(u) \cap \psi(v) \neq \emptyset$.
If $\{u,v\}= e \in E$, then $c_e \in \varphi(u) \cap \varphi(v) \cap U_E \subseteq \psi(u) \cap \psi(v)$, since we obtained the mapping $\psi$ from $\varphi$ by only removing colors not in $U_E$.
If $\{u,v\} \not \in E$ then $\psi(u) \cap \psi(v) = \emptyset$ follows directly from the relations $\psi(u) \subseteq \varphi(u)$ and $\psi(v) \subseteq \varphi(v)$, using the fact that $\varphi(u) \cap \varphi(v) = \emptyset$.
This shows that $(U,\psi)$ is an intersection representation of $G$ and, since $|\psi'(v)| = b(v)$ for all $v\in V$, completes the proof of the claim.
\end{proof}

Since $|U_E| = |E|$, to complete the proof that $|E| + b(V) - \nu(G,b)$ is a lower bound for $|U|$, we need to show that $|U\setminus U_E|\ge b(V)-\nu(G,b)$, or, equivalently, that $b(V)\le |U\setminus U_E|+\nu(G,b)$.
Consider the function $x:E\to \mathbb{Z}_+$ defined by setting
\[
    x(e) = |\varphi'(u)\cap \varphi'(v)|, \text{~for each edge~} e= \{u,v\}\in E\,.\]
We claim that $x$ is a $b$-matching of $G$.
Let $v$ be a vertex of $G$.
Let $E_v$ be the set of edges in $G$ that are incident with $v$ and let $N(v)$ be the set of neighbors of $v$ in $G$.
Recall that by the assumption that the graph is triangle-free, each color $c\in U\setminus U_E$ appears in at most two of the color sets $\varphi'(u)$, $u\in V$.
Therefore, for each vertex $v\in V$, the sets $\varphi'(u)\cap \varphi'(v)$, $u\in N(v)$, are pairwise disjoint.
By \cref{claim-varphi-b}, this implies that 
\[\sum_{e\in E_v}x(e) = \sum_{u\in N(v)}|\varphi'(u)\cap \varphi'(v)| = \left|\varphi'(v)\cap \left(\bigcup_{u\in N(v)}\varphi'(u)\right)\right|\le |\varphi'(v)| = b(v)\,,\]
and, hence, $x$ is indeed a $b$-matching of $D$.
For each $v\in V$, let us denote by $C_v$ the set of \emph{private} colors of $v$, that is, those colors $c\in \varphi'(v)$ such that $c\not\in \varphi'(u)$ for all $u\in V\setminus\{v\}$.
Using again \cref{claim-varphi-b}, we have
\begin{eqnarray*}
b(V) &=&\sum_{v\in V}b(v) 
 = \sum_{v\in V}|\varphi'(v)| %
= \sum_{v\in V}\left(|C_v|+|\varphi'(v)\setminus C_v|\right)\\
&=&  \sum_{v\in V}|C_v|+\sum_{v\in V}\sum_{u\in N(v)}|\varphi'(u)\cap\varphi'(v)|
=\sum_{v\in V}|C_v|+2\sum_{e\in E}x(e)\\
&=& 
\left(\sum_{v\in V}|C_v|+\sum_{e\in E}x(e)\right)+\sum_{e\in E}x(e) %
\le |U\setminus U_E|+\nu(G,b)\,,
\end{eqnarray*}
where the last inequality follows from the fact that each color in $U\setminus U_E$ is counted exactly once in the sum $\sum_{v\in V}|C_v|+\sum_{e\in E}x(e)$ and that $\sum_{e\in E}x(e)\le \nu(G,b)$, since $x$ is a $b$-matching in $D$.

Finally, we observe that an optimal $\ell$-constrained
intersection representation $(U,\varphi)$ of $G$ such that $|\varphi(v)|\le \max\{\ell(v),\deg(v)\}$, for all $v\in V$, can be computed in polynomial time. 
First, we compute the capacity function $b$ according to the definition.
Including also the time for the computation of the vertex degrees, this can be done in time $\O(|V|+|E|+|V| \cdot \log L)$.
Then we compute a maximum weight $b$-matching $x$ in $G$.
By \cref{max-weight-b-matching}, this can be done in time $\O\left(\min\{B \cdot |V|^2, |E|^2  \cdot \log |V| \cdot \log B\}\right)$, where $B = 1+\max_{v\in V}b(v)$.
Finally, we use \cref{constructing-U-from-x} to compute an intersection representation $(U,\varphi)$ of $(G,\preceq)$ with cardinality $|E|+b(V)-\nu(G,b)$.
This can be done in time proportional to the total size of this representation, which is, {by \cref{vertex-bound},} 
$\sum_{v\in V}|\varphi(v)| = \sum_{v\in V}{(b(v)+\deg(v))} \le \sum_{v\in V}{\max\{\ell(v),\deg(v)\}} = {\O(\sum_{v\in V}\ell(v)+|E|)}$.
We conclude that the overall time complexity is {$\O(|V|+|E|+{|V|\log L}+\min\{B \cdot |V|^2, |E|^2  \cdot \log |V| \cdot \log B\}+\sum_{v\in V}\ell(v)+|E|)$,
which, considering that $B = \O(L)$, simplifies to $\O(|V|\log L+\sum_{v\in V}\ell(v)+\min\{L \cdot |V|^2, |E|^2  \cdot \log |V| \cdot \log L\})$.}
\end{proof}

\subsection{The intersection number of triangle-free partially ordered graphs}

In this section we combine \cref{fromPO-to-LBtri-free-graphs,triangle-free-l-constrainedIN} into \cref{triangle-free-preceqIN}.

\trianglegreepreceqin*

\begin{proof}
We prove the statement by analyzing the  following algorithm.
\begin{enumerate}
    \item\label{demand'} Using \cref{fromPO-to-LBtri-free-graphs}, compute the $\alpha$-ranking $\ell:V\to \mathbb{Z}_+$ of $(G,\preceq)$.
\item\label{Uvarphi'} Using \cref{triangle-free-l-constrainedIN}, compute an optimal $\ell$-constrained intersection representation $(U,\varphi)$ of~$G$ such that $|\varphi(v)|\le \max\{\ell(v),\deg(v)\}$ for all $v\in V$.
\item\label{Upsi'} Using \cref{fromPO-to-LBtri-free-graphs}, transform $(U,\varphi)$ to an optimal $\ell$-constrained intersection representation $(U,\psi)$ of $G$ that is also an optimal intersection representation of $(G,\preceq)$.
\item\label{return'} Return $(U,\psi)$.
\end{enumerate}

\paragraph{Proof of correctness.}
By \cref{fromPO-to-LBtri-free-graphs,fromPO-to-LBtri-free-graphs}, the algorithm indeed computes an optimal intersection representation of $(G,\preceq)$.
Recall that the $\alpha$-ranking $\ell$ computed in step~\ref{demand'} is given by \cref{definition-of-ell}, that is,
\begin{equation*}\label{definition-of-w}
\ell(v) = \begin{cases}
 \deg(v) & \textrm{if } v \in M\,,\\
  \max\left\{\deg(v), 1+\max\limits_{u\colon  u \prec v} \ell(u)\right\} &
  \textrm{if } v \not\in M\,.
\end{cases}
\end{equation*}
Hence, for each $v\in V$, we have $b(v) = \ell(v)-\deg(v) = \max\{\ell(v)-\deg(v),0\}$.
By \cref{triangle-free-l-constrainedIN}, we have $\IN(G,\ell) = |E| + b(V) - \nu(G,b)$.
Furthermore, by \cref{fromPO-to-LBtri-free-graphs}, we have $\IN(G, \preceq) = \IN(G,\ell)$. 
Hence, $\IN(G, \preceq) = |E| + b(V) - \nu(G,b)$.

\paragraph{Time complexity analysis.}
By \cref{fromPO-to-LBtri-free-graphs}, step~\ref{demand'} can be done in time $\O(|V|+|E|)$.

By \cref{triangle-free-l-constrainedIN}, step \ref{Uvarphi'} can be done in time $\O(|V|\log L+\sum_{v\in V}\ell(v)+\min\{L \cdot |V|^2, |E|^2  \cdot \log |V| \cdot \log L\})$, where $L = 1+\max_{v\in V}\ell(v)$.
Using the facts that $L = \O(|V|)$ and $\sum_{v\in V}\ell(v) = \O(|V|+|E|)$, the time complexity of step \ref{Uvarphi'} simplifies to $\O(|V| \cdot \log |V|+\min\{|V|^3, |E|^2  \cdot \log^2  \cdot |V|\})$.

Note that for each vertex $v\in V$, we have $b(v)\le \ell(v)\le 2|V|$ (by \cref{fromPO-to-LBtri-free-graphs}).
Hence, $|U| = \IN(G,\preceq)= |E| + b(V) - \nu(G,b)\le |E|+\sum_{v\in V}b(v)\le |E|+2|V|^2$. 
By \cref{fromPO-to-LBtri-free-graphs}, step \ref{Upsi'} can be done in time $\mathcal{O}(|U| \cdot |V|+|E| \cdot \max_{v\in V}|\varphi(v)|)$, which, using the facts that $\max_{v\in V}|\varphi(v)|\le \max\{L,\max_{v\in V}\deg(v)\} = \O(|V|)$ and that $|U|= \O(|V|^2)$, simplifies to $\mathcal{O}(|V|^3)$.

The time complexity of step \ref{return'} is dominated by that of step \ref{Upsi'}.

The total time complexity is bounded by $\O(|V|+|E|+|V| \cdot \log |V|+\min\{|V|^3, |E|^2  \cdot \log^2  \cdot |V|\}+|V|^3)$, which simplifies to the claimed time complexity bound of $\O(|V|^3)$.
\end{proof}

\section{Further results and open questions}

\subsection{The bipartite case}

In a bipartite graph $G$, the maximum size of a $b$-matching is equal to the minimum $b$-weight of a vertex cover (see, e.g.,~\cite{MR1956926}).
Using this, we can give an alternative expression of the result in 
\cref{triangle-free-Hamiltonian-DIN-main} in the case of \textsl{bipartite} Hamiltonian DAGs.

\begin{corollary}\label{bipartite-Hamiltonian-DIN}
Let $D = (V,A)$ be a bipartite DAG with a Hamiltonian path $P = (v_1,\ldots, v_n)$.
Let $w:V\to \mathbb{Z}_+$ be a vertex weight function on $D$ defined recursively along $P$ as follows:
$w(v_1) = \deg(v_1)$, and for all $i\in \{2,\ldots,n\}$, we set $w(v_i) = \max\{w(v_{i-1})+1,\deg(v_i)\}$.
Let $b:V\to \mathbb{Z}_+$ be a capacity function on the vertices of $D$ defined by setting $b(v) = w(v)-\deg(v)$ for all $v\in V$.
Then $\DIN(D) = |A|+\alpha(D,b)$.
\end{corollary}

\begin{proof}
The result follows immediately from 
\cref{triangle-free-Hamiltonian-DIN-main}. Let $G$ be the the underlying undirected graph of $D$. Let $x:E \to \mathbb{Z}_+$ a 
$b$-matching for $G$ with respect to the demand function $b$ and such that its weight is $\nu(G, b)$.
Let $C$ be a vertex cover for $G$ of minimum weight
with respect to the vertex weight function $b$.
Since the complement of an independent set is a vertex cover, we also have $b(C) = b(V) - \alpha(G,b)$.
Then, since $G$ is bipartite we have that 
$\nu(G, b) = b(C)$, which yields the desired result.
\end{proof}

\subsection{Constant factor approximations} 

Before this paper, the only class of graphs for which a constant factor approximation polynomial-time algorithm was known for computing the DIN was the class of arborescences (see~\cite{IWOCA2022}). 
As a consequence of our results, we can significantly widen the class of DAGs where the problem admits a constant approximation.
More specifically, the DIN can be approximated in polynomial time to a constant factor on DAGs with bounded chain cover number whose underlying graph has bounded chromatic number.
Given an undirected graph $G$, the \emph{chromatic number} of $G$ is the minimum integer $k$ such that the vertex set of $G$ can be expressed as a union of $k$ independent sets.
A \emph{chain} in a DAG $D$ is a sequence $(v_1, \ldots, v_k)$ of distinct vertices of $D$ such that for all  $i\in \{1,\ldots, k-1\}$ there is a path from $v_i$ to $v_{i+1}$ in $D$.
A \emph{chain cover} in a DAG $D$ is a collection of chains such that every vertex in $D$ belongs to exactly one of the chains.
The \emph{chain cover number} of $D$ is the minimum cardinality of a chain cover in $D$.

\begin{proposition}
For every two positive integers $c$ and $k$, the DIN can be approximated in time $\mathcal{O}(n^2)$ to a factor of $2c^2k^2$ on DAGs with $n$ vertices having chain cover number at most $c$ and whose underlying graphs have chromatic number at most $k$.
\end{proposition}

\begin{proof}
In the proof of~\cite[Lemma II.1]{MR4231959} and using Remark II.2, Liu et al.\ give a procedure for computing a directed intersection representation of an $n$-vertex DAG using at most $\left\lfloor 5(n+1)^2/8-(n+1)/4\right\rfloor \le n^2$ colors.
Each of the steps of the algorithm requires $\mathcal{O}(n^2)$ time.

Given a DAG $D = (V,A)$ with at least two vertices and chain cover number at most $c$ such that the chromatic number of $U(D)$ is at most $k$, we apply the mentioned algorithm of Liu et al.\ to $D$.
We claim that the obtained number of colors, which is at most $n^2$, exceeds $\DIN(D)$ by a factor of at most $2c^2k^2$.
Since $D$ admits a chain cover with at most $c$ chains, at least one of these chains contains at least $n/c$ vertices.
Fix a chain $C = (v_1,\ldots, v_p)$ in $D$ such that $p\ge n/c$. 
In any directed intersection representation $\varphi$ of $D$, the cardinalities of the color sets need to be strictly increasing along each chain, and in particular along $C$.
Furthermore, by assumption on $D$, the vertex set of $D$ can be partitioned into $k$ independent sets.
It follows that there exists an independent set $I$ in $D$ such that $I\subseteq V(C) = \{v_1,\ldots, v_p\}$ and $|I|\ge p/k\ge n/(ck)$.
Since for any two distinct vertices $u,v\in I$, the cardinalities of their color sets $\varphi(u)$ and $\varphi(v)$ are distinct, we obtain that these two sets are disjoint.
Consequently, the total number of colors used on vertices of $I$ is at least $\sum_{j = 1}^{|I|}j = \binom{|I|+1}{2}$, which is at least $|I|^2/2\ge n^2/(2c^2k^2)$.
Therefore, $\DIN(D)\ge n^2/(2c^2k^2)$ and the claimed approximation guarantee follows.
\end{proof}

\subsection{Open questions}

We have shown in \cref{properties-of-alpha-ranking} that $\IN(G,\ell)\le \IN(G, \preceq)$ holds for every partially ordered graph, where $\ell$ is the $\alpha$-ranking of $(G,\preceq)$, and that the converse inequality does not hold in general (\cref{prop:strict}).
Can the difference between these two parameters be 
arbitrarily large?
Furthermore, if this is the case, is the intersection number of the partially ordered graph $(G, \preceq)$ bounded from above by some function of the $\ell$-constrained intersection number of $G$?

Another general question suggested by this work is about the extent, in terms of graph classes, to which the relationships between the different intersection representations hold (see the diagram in \cref{Problems-relationships}). 
In particular, does the fact that the $\ell$-constrained intersection number of a graph $G$ generalizes the computation of intersection number over any partially ordered graph $(G, \preceq)$ hold beyond the class of diamond-free graphs?

Another open question is whether a result analogous to \cref{triangle-free-preceqIN}, giving a polynomial-time algorithm for computing the intersection number of a partially ordered triangle-free graph, could be obtained for the more general class of partially ordered diamond-free graphs.
We were able to show, in
\cref{thm:diamond-free},
that for any partially ordered diamond-free graph $(G,\preceq)$, the computation of an optimal intersection representation for $(G, \preceq)$ can be reduced to the computation of an optimal $\ell$-constrained intersection representation of $G$ for some polynomially computable demand function $\ell$.
However, it is open whether an optimal $\ell$-constrained intersection representation of a diamond-free graph $G$ can be computed in polynomial time.

In this regard, let us remark that the approach of reducing the problem to the appropriate generalization of the $b$-matching problem does not seem to work, at least not as simply as in the triangle-free case.
The appropriate generalization of the $b$-matching problem would be the following problem:
Given a graph $G$ and a function $b:V(G)\to \mathbb{Z}_+$, find a function $x:\mathcal{C}(G)\to \mathbb{Z}_+$, where $\mathcal{C}(G)$ is the set of maximal cliques of $G$, such that for all $v\in V(G)$, the sum of the values $x(C)$ over all maximal cliques $C$ containing $v$ does not exceed $b(v)$, and the sum $\sum_{C\in \mathcal{C}(G)}x(C)$ is maximized.
It turns out that this problem is \textsf{NP}-hard when restricted to diamond-free graphs, even in the case of unit capacity function.
While the proof is similar to the proof of \textsf{NP}-completeness of the triangle packing problem in line graphs due to Guruswami et al.~\cite{DBLP:conf/wg/GuruswamiRCCW98}, we include it for the sake of completeness.

\begin{theorem}
Given a diamond-free graph $G$ and an integer $k$, it is \textsf{NP}-complete to determine if $G$ admits $k$ pairwise disjoint maximal cliques.
\end{theorem}

\begin{proof}
The problem is clearly in \textsf{NP}.
To prove \textsf{NP}-completeness, we make a reduction from  \textsc{Independent Set in Triangle-Free Cubic Graphs},\footnote{A graph $G$ is \emph{cubic} if every vertex of $G$ is incident with precisely three edges.} which is the following \textsf{NP}-complete problem (see~\cite{uehara1996np,DBLP:conf/wg/GuruswamiRCCW98}):
Given a triangle-free cubic graph $G = (V,E)$ and an integer $k$, does $G$ contain an independent set $I$ such that $|I|\ge k$?
Let $(G,k)$ be an input instance to \textsc{Independent Set in Triangle-Free Cubic Graphs}.
Let $H$ be the line graph of $G$, that is, the graph with vertex set $E$, in which two distinct vertices are adjacent if and only if the corresponding edges of $G$ share an endpoint.
Then, $H$ is diamond-free (this follows, e.g., from~\cite[Theorem 4]{zbMATH06766595}).
Since $G$ is cubic and triangle-free, the maximal cliques in $H$ correspond precisely to the sets of edges incident with a fixed vertex of~$G$; consequently, the maximal cliques in $H$ are precisely its triangles.
Furthermore, $G$ has an independent set of size $k$ if and only if $H$ has $k$ pairwise disjoint triangles (see~\cite[Lemma 2]{DBLP:conf/wg/GuruswamiRCCW98}).
Therefore, $G$ has an independent set of cardinality~$k$ if and only if the diamond-free graph $H$ has $k$ pairwise disjoint maximal cliques, and the claimed \textsf{NP}-completeness follows.
\end{proof}

\section*{Acknowledgments}
We would like to thank Andrea Caucchiolo for several 
insightful discussions we had on some of the results contained in this paper.

\printbibliography

\end{document}